%% file: paper.tex
\documentclass[epsf,11pt,times,psfrag]{article}

\usepackage{multirow,amssymb,amsmath,graphicx,epsfig,graphics,geometry,setspace}
\usepackage{wasysym,hyperref,textcomp}
\newcommand{\ignore}[1]{{}}

\newenvironment{proof}{\par\noindent{\bf Proof:}}{\mbox{}\hfill$\Box$\\}

\newtheorem{theorem}{Theorem}[section]
\newtheorem{thm}{Theorem}[section]

\newtheorem{lem}{Lemma}[section]

\newtheorem{cor}{Corollary}[section]

\begin{document}
\title{Efficient cache oblivious algorithms for randomized\\
divide-and-conquer on the multicore model}
\author{
    Neeraj Sharma\\
    Mentor Graphics (INDIA) Pvt. Ltd.\\
    Noida, U.P. 201301, India \\
    neeraj\_sharma@mentor.com
  \and
    Sandeep Sen\\
    Department of CSE\\
    I.I.T. Delhi, New Delhi 110016, India\\
    ssen@cse.iitd.ac.in
}
%\affiliation{
%XYZ
%}  

%\affiliation{
%IBM India Research Lab, New Delhi\\
%\email{}
%}  

\maketitle
\begin{abstract}
In this paper we present randomized algorithms for sorting and convex hull 
that achieves optimal performance (for speed-up and cache misses) on the
multicore model with private cache model. Our algorithms are cache oblivious
and generalize the randomized divide and 
conquer strategy given by Reischuk \cite{refreischuk} and Reif and 
Sen \cite{3dhull}. 
Although the approach yielded optimal speed-up in the PRAM model, we require
additional techniques to optimize cache-misses in an oblivious setting.
Let $p,n,M,B$ respectively denote number of processors, problem size, the
size of individual processor cache memory and block size respectively, then
we obtain expected parallel running time 
$O(\frac{n}{p}\log n + \log n \log \log n)$
 with expected $O(\frac{n}{B} \log_M n)$ cache misses for sorting $n$
keys and constructing convex hull of $n$ points.
 %The running times are optimal if $p \leq (n/ \log \log n)$. 
For $p \leq \frac{n}{\log \log n}$, and under the tall-cache
assumption $M \geq B^2$, both speed-up and cache-misses are the
best possible. 
Since the input-size $n\geq Mp$, under a very mild assumption 
$M \geq \log\log n$, $p \leq \frac{n}{\log\log n}$, so 
in all realistic scenarios, our algorithm will have optimal time and 
cache misses with high probability.
Although similar results have been obtained recently 
for sorting \cite{resource} , we feel that our
approach is simpler and general and we apply it to obtain an
 algorithm for 3D convex hulls 
with similar bounds. 

We also present a simple 
randomized processor allocation technique
without the explicit knowledge of the number of processors 
that is likely to find additional applications in resource oblivious 
environments. 

\end{abstract}
\vfill{SPAA'12,} {June 25--27, 2012, Pittsburgh, Pennsylvania, USA.} 
\newpage
\input{intro}
\input{basic}
\input{sorting}
\input{convex}

\input{procalloc}
\input{reference}
\end{document}

%% file: intro.tex
%\chapter{Multi-Core Model}
\section{Introduction}
The private-cache multicore model and the closely related 
Parallel External Memory (PEM) model 
combines several features of parallel computing models like PRAM and the 
memory hierarchy
issues captured by the External Memory Models. The goal is to capture the 
relevant aspects of a large scale multiprocessing environment, whose numerous 
parameters may be unknown to the algorithm designer. Although, it is not
intuitive, this last requirement can be tackled using the strategy called
{\em cache obliviousness} or more generally {\em resource obliviousness}.  

These multiprocessing models consists of $p$ processors (or cores) 
each having a private cache of size $M$. There is a large global memory which 
acts as a shared memory.
The processors communicate with each other only through shared memory. 
Initially the input is present in the global memory stored in form of blocks 
of size $B$.
So an input of size $n$ uses $n/B$ blocks in the global memory. All the 
transfers to and from memory (or cache) are done in blocks of size $B$.
Whenever a core needs some data, say $x$, if the data is already present in its 
private cache then no cost is incurred but if it is not
present in the cache then it copies this block from the shared memory to its 
private cache. Similarly when a core wants to write a block, first that block 
is moved into its private cache then 
those changes can be updated in shared memory at that time or later using 
some write-back policies. When a core modifies a block in its cache, 
then all the copies of that block which 
are present in the cache of other cores are invalidated and for any future 
access to that block it needs to be read again from the shared memory.
To obtain full parallelism we need to fill the cache of each core at least 
once, i.e.\textit{ $n \geq M \cdot p $.}\\
\par Thus we have two types of cache related-cost incurred in this 
model\footnote{These terms have been previously 
defined in \cite{resource}}.
\begin{enumerate}
 \item \textbf{Cache misses : } This is the standard type of cache miss which also occurs in sequential computation. Whenever any core need some 
data which is not present in its cache,
that block is copied from main memory to its cache. This is treated as one cache miss. If some block $B_1$ is evicted from cache because
of cache replacement policy and processor again need this $B_1$ block, 
it will again cost one cache miss to access this block (also known as
{\em capacity} misses).
\item \textbf{Block misses : } These kind of misses occur only in the case 
of concurrent writing. It doesn't have an analogue 
in sequential computation and requires extra care
in parallel computation.
Suppose two cores $C_1$ and $C_2$ are accessing the same block $B_1$. 
If $C_1$ modifies the contents of this block then the copy of $B_1$ present in
$C_2$'s cache is invalidated and it has to be read again which 
will be counted as one cache miss. Note that block misses increase
when a larger number of cores are accessing the same block and one core modifies 
it. 
So the block misses occur when multiple processors are 
accessing the same block at same time and at least one of them is writing. 
An algorithm designer will try to avoid this situation as much as possible.
\end{enumerate}
Both concurrent reads and concurrent writes are allowed in this model.
\\
 \textbf{Concurrent Reads :} 
If $x$ cores are reading the same block, every core will incur one 
cache miss and the total cache cost will be $x$ (unless some processor
writes on this block that will be treated as a block miss).
\\
\textbf{Concurrent Writes :} We have two cases: 
\\
1. When $x$ cores are reading one block and one core writes to that 
block at same time. All these cores will have to read this block again
from shared memory, which will lead to a total of $x$ cache misses. 
\\
2. If there are $x$ cores and all of them want to write to the 
same block, this block will migrate across these $x$ cores
to satisfy their write requests. Thus the $i^{th}$ core in sequence will 
have to wait for time equal to $i$ cache misses before it can
complete its write operation. So the total cache misses across all cores 
can be $\sum_{i=1}^{x-1} i = \Omega ( x^2 )$. 

In the worst case, any core will incur $B$ cache misses to 
read or write on any block. So, while distributing problems between different
cores, to save cache cost, we will try to avoid block misses.

%Note that the above model of dealing with concurrent accesses is weaker than
%the one described in \cite{four} that allows a {\em binary tree} merging
%of requests to reduce the parallel I/O complexity of concurrent access
%to a block to $O(\log B)$ from a potential $\Omega ( B)$. 
The performance of any algorithm is characterized by two parameters.
 \\
\noindent \textit{Cache misses} : The total number of cache misses and block 
misses incurred during algorithm. We will analyze the total cache cost wherever 
any block miss occurs.\\
\textit{Critical path} : It is the maximum time taken by any processor in the
overall algorithm.
To decrease the critical path, we have to ensure equitable work-load 
distribution.\\
{\em Note} : Apart from its close relation to parallel time, this also affects
the overall cache misses in a multiprogramming scheduling algorithm based
on {\em work-stealing} \cite{refprefix}.

Further, based on our {\em a priori} knowledge of multicore parameters 
$(M,B,p)$ we have the following variations
 \\
\textbf{Cache Aware : }
The algorithm designer makes use of the values of 
multicore parameters $(M,B,p)$ - for example, 
we can divide the problem according to values of $M$ and $B$ to 
better utilize the cache.
\\
\textbf{Cache Oblivious : }
In this model we know the number of processors while designing algorithm but 
the values of $M$ and $B$ are unknown. However, 
the analysis can make use of these parameters and our goal is to match the
performance of Cache aware algorithms.
\\ 
\textbf{Resource Oblivious Model : }
In this model, even the number of processors is 
unknown to the algorithm designer. Usually the tasks are shared 
by processors using some on-line strategy
like work-stealing methods \cite{cilk}
\subsection{Previous Related Work}
Our results make use of Reischuk's \cite{refreischuk} 
parallel randomized sorting algorithm that sorts
 $n$ elements using CREW PRAM processors in $O(\log n)$ time with high
 probability. For the external memory model, 
Aggarwal and Vitter \cite{vitter} 
designed a version of merge sort, 
 that uses a maximum $O(\frac{N}{B}\log_{M/B}N/B)$ I/O's and this
is optimal. 
For the cache oblivious model, Frigo et al. \cite{reftranspose} 
presented a sequential cache-oblivious sorting algorithm called 
Funnel sort which can 
sort $n$ keys in $O(n\log n)$ time and $O(\frac{n}{B}\log_M n)$ cache
misses\footnote{For tall cache, since $\log_M n$ is $O(\log_{M/B} n/B )$, 
we will
use the shorter expression}. Note that both time and cache misses achieved 
by this algorithm are optimal.

Arge et al. \cite{four} formalized the PEM model
and presented a merge sort algorithm based on \cite{cole} that runs in 
$O(\log n)$ time and has optimal cache misses. Note that their
 model is both cache aware and processor aware.
This algorithm is very similar to merge sort implemented on $\cite{vitter}$. 
They proved these results assuming that $n\geq pB^2$ by using a $d$-way merge sort which partitions the input into $d$ subsets,
then sorts each subset recursively and merges them. To
achieve optimal I/O complexity and parallel speedup, the
sorted subsets are sampled and these sample sets are merged
first. Then using the values of $M$ and $B$, a suitable value of $d$ is chosen to achieve the claimed bounds.

 Blelloch et al.~\cite{refprefix} presented a resource oblivious 
distribution sort algorithm 
that has $O(\log^2 n)$ critical path-length and incurs sub-optimal cache cost
in the private-cache multicore model.
Their distribution sort uses merging to divide the input 
into contiguous subsets and sorts them recursively. 
For their randomized 
version, they proved an expected $O(\log^{1.5}n)$ depth.
The potential bottleneck for reducing the depth further 
was their technique for using merging to partition into buckets.
The algorithm given in \cite{sixteen}
is designed for a BSP-style version of a cache aware, multi-level multicore. It
uses a different collection of parameters, and so it is difficult to 
compare it directly with the previous results.

Recently Cole and Ramachandran \cite{resource} presented a new optimal merge 
sort algorithm (SPMS)  for resource oblivious multicore model. 
This algorithm sorts $n$ keys in $O(\frac{n}{p} \log n + \log n \log \log n)$ time using $n$ processors with
optimal number of cache misses on resource oblivious model assuming 
$n\geq Mp$ and $M\geq B^2\log B \log \log B.$
It works in $O(\log \log n)$ depth where each stage requires $O(\log n)$ time.
The authors addressed a general computational paradigm called
{\em Balanced Partitioning Trees} and designed a
a resource-oblivious 
priority work scheduler based on work-stealing to attain the above bounds. 
It follows the same approach as merge sort where we splits the input into 
subsets, sorts
the subsets recursively and then merges them (to achieve the desired bounds,
the authors used multi-way merging).

As sorting is basic problem in algorithms similarly Convex hull problem is a basic problem in Computational geometry.
There is a close relation between sorting and convex hull problem. It has been proved that sorting can be reduced to
convex hull problem. So it means we can't solve convex hull problem faster than sorting. 
On a single processor convex hull problem can be solved easily by reducing it to sorting. For 2-D convex hull problem, on single processor,
there exists a output-sensitive cache oblivious algorithm which achieve optimal time and cache misses \cite{2dhull}. By saying output sensitive hull,
we means that total time is proportional to number of points in output of convex hull. 
This algorithm is a modification of Chan's output sensitive algorithm 
\cite{chan}.
On CREW model, there exists an algorithm which solve 3-D convex hull problem in expected $O(\log n)$ time using $n$ processors \cite{3dhull}.\\
In the table below, we summarize the results discussed above where
 \~{O} is used to represent expected bound.
\begin{center}
\begin{tabular}[c]{|c|c|c|c|l|}
\multicolumn{5}{c}{\multirow{2}*{\textit{To sort $n$ elements using $p$ processors}}} \\ 
\multicolumn{5}{c}{\multirow{2}*{}} \\ \hline
\multirow{2}*{Model} & \multirow{2}*{Time} &  \multirow{2}*{Cache cost}    &  \multirow{2}*{Condition} & \multirow{2}*{Source}    \\
& & &  & \\ \hline
\multirow{1}*{CRCW} & \multirow{2}*{ \~{O}$(\frac{n}{p}\log n)$} & \multirow{2}*{$-$} & \multirow{2}*{$n\geq p$} & \multirow{2}*{Reischuk \cite{refreischuk}}  \\
\multirow{1}*{PRAM}& & & &   \\ \hline
\multirow{1}*{EREW} & \multirow{2}*{ O$(\frac{n}{p}\log n)$} & \multirow{2}*{$-$} & \multirow{2}*{$n\geq p$} & \multirow{2}*{Cole \cite{cole}}  \\
\multirow{1}*{PRAM}& & & &   \\ \hline
\multirow{1}*{Cache} & \multirow{2}*{ O$(\frac{n}{p}\log n)$} & \multirow{2}*{$O(\frac{n}{B}\log_{\frac{M}{B}} \frac{n}{B})$} & \multirow{1}*{$n\geq B^2p$} & \multirow{2}*{Goodrich \cite{four}}  \\
\multirow{1}*{Aware}& & &\multirow{1}*{$M\geq B^2$} &   \\ \hline
\multirow{1}*{Cache} & \multirow{2}*{ O$(\frac{n}{p}\log^2 n)$} & \multirow{2}*{$O(\frac{n}{B}\log \frac{n}{B})$} & \multirow{1}*{$n\geq Mp$} & \multirow{2}*{Ramachandran \cite{five}}  \\
\multirow{1}*{Oblivious}& & &\multirow{1}*{$M\geq B^2$} &   \\ \hline
\multirow{1}*{Cache} & \multirow{2}*{ Depth$^3$
 =  O$(\log^2 n)$} & \multirow{2}*{$O(\frac{n}{B}\log_{\frac{M}{B}} \frac{n}{B})$} & \multirow{2}*{$M=\Omega (B^2)$} & \multirow{2}*{Blelloch \cite{refprefix}}  \\
\multirow{1}*{Oblivious}& & & &   \\ \hline
\multirow{1}*{Cache} & \multirow{2}*{ Depth$^3$ = \~{O}$(\log^{1.5} n)$} & \multirow{2}*{\~{O}$(\frac{n}{B}\log_{\frac{M}{B}} \frac{n}{B})$} & \multirow{2}*{$M=\Omega (B^2)$} & \multirow{2}*{Blelloch \cite{refprefix}}  \\
\multirow{1}*{Oblivious}& & & &   \\ \hline
\multirow{1}*{Cache} & \multirow{1}*{ \~{O}$(\frac{n}{p}\log n +$} & \multirow{2}*{\~{O}$(\frac{n}{B}\log_{\frac{M}{B}} \frac{n}{B})$} & \multirow{1}*{$n\geq Mp$} & \multirow{2}*{Our Results} \\
\multirow{1}*{Oblivious}&\multirow{1}*{ $\log n \log \log n)$} & & \multirow{1}*{$M\geq B^2B^{2/31}$} &   \\ \hline
\multirow{1}*{Resource} & \multirow{1}*{ O$(\frac{n}{p}\log n +$} & \multirow{2}*{O$(\frac{n}{B}\log_{\frac{M}{B}} \frac{n}{B})$} & \multirow{1}*{$n\geq Mp$, $M\geq$} & \multirow{2}*{ \cite{resource}} \\
\multirow{1}*{Oblivious}&\multirow{1}*{ $\log n \log \log n)$} &  & \multirow{1}*{$B^2\log B \log \log B$}&   \\ \hline

\end{tabular}
\end{center}
\footnote{ It didn't account for extra costs like block misses}
\subsection{Our Work}
In this paper, we have presented a randomized distributed sorting algorithm on 
cache-oblivious multicore model. 
Our basic approach is similar to Reischuk \cite{refreischuk}, 
but to bound cache misses, we 
had to modify it significantly. First, we sample some elements from input, 
sort them using a brute-force method and 
use these elements to divide the input into disjoint buckets. 
To do this partitioning with minimal cache misses in a cache-oblivious
fashion, we do it in two phases. Roughly speaking, we divide the $n$ 
input keys into $\sqrt{n}$
buckets by successive partitioning into $n^{1/4}$ size buckets using a
common procedure. This partitioning procedure, which is the crux of our
distribution sort, in turn invokes an efficient merging procedure to attain the
final bounds.

 We are able to sort $n$ elements in expected $O(\frac{n}{p}\log n + \log n \log \log n)$
 time with expected $O(\frac{n}{B} \log_M n)$ cache misses. These cache misses match the optimal cache bound 
$O(\frac{n}{B}\log_{\frac{M}{B}} \frac{n}{B})$ for tall cache i.e. if $M \geq 
B^2.$
  So the cache cost is optimal but time is optimal only for 
$p \leq (n/ \log \log n)$. From the natural condition that
$n\geq Mp$, if $M\geq \log \log n$ it implies that 
$p \leq (n/ \log \log n)$ and our algorithm will have both time and cache 
misses optimal.
Our bounds for sorting match that of Cole and Ramachandran \cite{resource} 
in cache-misses and depth and we can obtain matching performance in the 
cache-oblivious PEM model. We also present a simple technique
for processor-obliviousness where the algorithm need not
have any prior knowledge of the number of processors and the processors
can generate their ids on the fly. Our approach is fundamentally different
from \cite{resource,refprefix} that designs a scheduler to map tasks to
processors in a resource-oblivious fashion.  

 In practice, size of cache is $32$KB or $64$KB and size of block is $256$B  it means $M\geq \log \log n$ condition is satisfied for input $\leq 2^{(2^{8000})}.$
So our algorithm have optimal time and cache cost with high probability for input of size $\leq$ $2^{(2^{8000})}$.

Moreover, our approach yields a general framework for 
randomized divide-and-conquer that has other applications. In the latter
part of this paper, we 
present a parallel multicore algorithm for constructing 
convex hulls of point sets. This is based on the approach of 
of Reif and Sen~\cite{3dhull} but we follow a simpler description given in 
\cite{ger} for planar hulls. We are able to apply a number of 
subroutines developed for the
parallel sorting algorithm for the construction of the 
convex hull and achieve a bound similar to sorting. Since it is
known that Voronoi diagrams can be reduced to three 
dimensional convex hulls,
we obtain identical results as given above for sorting.

%% file: basic.tex
\section{Some basic algorithms }

The reader may recall that 
we are designing these algorithms without using the parameters $M$ and $B$.
First we mention 
some basic algorithms which we will need for our sorting algorithm.
We will use $T$ to represent total time (total operations), $\bar{T}$ to represent expected time, $T^{\prime \prime}$ to represent parallel time,
$\bar{T}^{\prime \prime}$ to represent expected parallel time and 
$Q$ to represent total cache misses. \\
Some times processors may have to read or write in contiguous locations 
which can lead to block misses. So we bound total cache misses as follows :\\
\textbf{Contiguous Reads and Writes.} 
We have $p$ cores and a total of $n$ keys in the main memory stored in $n/B$ 
blocks and every processor has to read equal number of keys. 
So the total cache misses for every
core will be $\lceil\frac{n}{pB}\rceil.$ Since a core may start or 
end reading in the middle of any block, the total cache misses across all cores
should be $\leq n/B + 2p$ which is $O(n/B)$ 
given our assumption that $(n\geq Mp)$.
Similarly, when $p$ cores have to write a total of $n$ keys in $n$ 
contiguous locations and every processor has to write an equal number of 
keys $n/p$, then there will not be 
more than two block misses per processor for $n\geq Bp$.

We now state results for some basic problems used throughout this paper.
The associated algorithms are based on standard approaches and the bounds 
assuming that $n \geq Mp$.
\begin{lem}\label{maximum}
Using $p$ cores, maximum of $n$ keys can be found in time O$(n/p+ \log p)$ with total $O(\lceil n/B \rceil)$ cache misses. 
\end{lem}
\begin{proof}
Assume input has been divided into $p$ contiguous chunks each of size $n/p$. Every core is assigned a task to get a maximum element from one chunk.
This will take $\lceil n/p \rceil$ time and $n/B+p$ cache misses.\\
We get $p$ elements. Use $p/2$ cores, where every core will read two keys and reject one key. This step will take $O(1)$ time and $p$ cache misses.  
Repeat the same procedure using $p/4$ cores and so on until we are left with one input key.
Total time for whole algorithm will be $O(n/p+\log p)$ and total cache misses across all cores will be $O(n/B+p)$ which can be bounded by $O(n/B)$ if $n\geq Bp.$
\end{proof}
\begin{lem}\label{sum}
Using $p$ cores, $n$ keys can be added in time O$(n/p+ \log p)$ with total $O(n/B)$ cache misses.
\end{lem}
\begin{proof}
This task can be done using the same approach as used in Lemma \ref{sum} within same bounds.
\end{proof}
\begin{lem}\label{prefix}
Using $p$ cores, prefix computation on $n$ keys can be done in time O$(n/p+ \log p)$ with a total of $O(n/B)$ cache misses.
\end{lem}
\textit{Input} : $(k_1,k_2,\ldots,k_n)$\\
\textit{Output} : $(k_1,k_1+k_2,k_1+k_2+k_3,\ldots, \displaystyle\sum_{i=1}^{{n}}k_i)$\\
\begin{proof}
We will use an algorithm which does prefix computation on $n$ keys in $O(n)$ time and $O(n/B)$ cache misses using one processor. \cite{refprefix}.
Due to its Divide and Conquer nature, this algorithm can be easily adapted to run on more than one processor.
To bound the block misses, we need to makes some small changes.\\
We are given an input array $A$ of length $n$, a temporary array $S$ of length $n-1$, and an output
array $R$ of length $n$. Say K=$(n/p)$.
Consider a balanced binary tree laid over input as shown in figure \ref{imgpre4} (for $n=8$). This balanced tree will be of size $n-1$ and will be saved in array
$S$. To make memory operations more cache efficient and avoid block misses (explained later), save this tree in infix order in $S$ as shown in
figure \ref{imgpre4}.

\begin{figure}[ht]
\begin{center}
\includegraphics[scale=0.60]{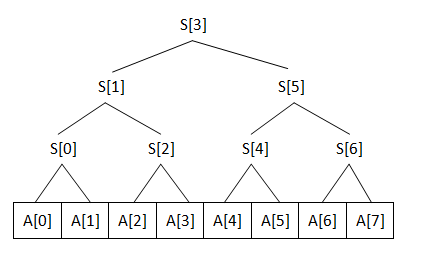}
\end{center} 
\caption{For $n$ = $8,$ $A$ is input array and $S$ is temporary array }
\label{imgpre4}
\end{figure}

\noindent{}This algorithm will work in two phases. In first phase, we will calculate the sum of the left subtree for each node and in second phase we will push this sum down the tree to calculate 
prefix sum. The recursive codes for the two phases are shown below :
\begin{enumerate}
\item[Phase 1:] $phase_1(i, size)$. $A$ is input array and $S$ is temporary array where output of this phase $1$ is stored.
\begin{enumerate}
\item if $size = 1$ then return $A[i]$.
\item $S[i+size/2]=phase_1(i, size/2)$
\item return $S[i+size/2] + phases_1(i+size/2, size/2)$
\end{enumerate}
\item[Phase 2:] $phase_2(i,size,s)$. Output of this phase is final output and is stored in array $R$.
\begin{enumerate}
\item if $size = 1$ then set $R[i]=s$ and return.
\item $phase_2(i,size/2,s)$
\item $phase_2(i+n/2,size/2,s+S[i+n/2])$
\end{enumerate}
\end{enumerate}
\textbf{Processor Allocation :} After every level, no of processes get multiplied by two. So we know how many processes we 
would have at any time which means calls can be distributed easily to processors using processors numbers.
Processors do not need any other information to execute these calls. Thus we do not have to spend any extra time in processor allocation.\\
\textbf{Analysis :}
Both phases use the same approach and same time and cache misses. We will analyze $phase_1$ and same result follows for $phase_2$ also.
Assume K=$n/p$. Each call to $phase_1$ is divided into two calls of half sizes after constant computation and these both calls can be executed on 
different processors in parallel at same time. But if size of problem is $\leq K$ then we have enough subproblems such that every processor can solve 
one problem individually and sequentially. So if size $\leq K$, all processes created by this call will be executed on same processor. This modification helps us to
bound total number of cache misses when size becomes below $M$, total cache misses will be bounded by $M$ if the problem is solved on single core.\\
\textbf{Block Misses :}\\ 
Phase 1: Block misses occur when more than one processors are trying to write on same block. All written operations in phase 1 are done only on $S$ array.
Because of the way binary tree is stored in $S$, we are able to bound block misses. During $phase_1$ we are basically traversing tree
stored in $S$ from top to bottom (Figure \ref{imgpre4}) where very node represents one call and two children represents two newly created sub-calls. We
write one value on array $S$ and create two new sub-calls to left and right subtree. All nodes of tree which are present at same level are being executed
in parallel and each node is writing on one location in $S$.
Distance between any two written positions will be equal to size of nodes on that level. So block misses are zero till problem size is $> B$. 
And when problem size goes $\leq n/p$, it is executed on single 
processor. In this subtree every processor is going to write contiguous elements (more than $M$).
And when a problem is being solved on one processor than all calls are executed in depth first order
which means that every core will write its elements in order from left to right in $S$ which gives us zero block misses.
So block misses are zero when problem size is greater than  B or less than $n/p$. On the whole, we will have zero block misses.\\
Phase 2 : In $phase_2$ also we are traversing through a binary tree in which only leaves are taking part in writing. So similar to $phase_1$, block misses
are zero if $n\geq Mp.$\\
\textbf{Time and Cache misses :}
Say $T(n,p)$ represents total time, $T^{\prime \prime}(n,p)$ represents parallel time and $Q(n,p)$ represents total cache misses taken by $phase_1$ for input of size $n$ using $p$ cores.
\begin{equation*}
\mbox{$T(n,1)$ = }
\begin{cases}
1 &\text{if $n=1$}\\
O(1) + 2T(n/2,1) & \text{otherwise}
\end{cases}
\end{equation*}
\begin{equation*}
\mbox{$Q(n,1)$ = }
\begin{cases}
n/B &\text{if $n \leq M $}\\
O(1) + 2Q(n/2,1) & \text{otherwise}
\end{cases}
\end{equation*}
Solving these both equations, we get $T(n,1)=O(n)$ and $Q(n,1)=O(n/B).$\\
Say $K=n/p$, where $K\geq M$
\begin{equation*}
\mbox{$T^{\prime \prime}(n,p)$ = }
\begin{cases}
n &\text{if $n\leq K$}\\
O(1) + T^{\prime \prime}(n/2,p/2) & \text{otherwise}
\end{cases}
\end{equation*}
\begin{equation*}
\mbox{$Q(n,p)$ = }
\begin{cases}
n/B &\text{if $n \leq K $}\\
O(1) + 2Q(n/2,p/2) & \text{otherwise}
\end{cases}
\end{equation*}
Solving these both equations, we get $T^{\prime \prime}(n,p)=O(n/p +\log p)$ and $Q(n,p)=O(n/B).$\\
This concludes the proof of Theorem \ref{prefix}.
\end{proof}
\begin{lem}\label{transpose}
Using $p$ cores, a matrix of size $m\times n$ can be transposed in time O$(mn/p+ \log p)$ with total $O(mn/B)$ cache misses given $mn\geq Mp$.
\end{lem}
\begin{proof} As given in $[\cite{reftranspose}]$, Using one processor, we can transpose any matrix of size $m\times n$ in $O(mn)$ time and $O(mn/B)$ cache misses.
We will implement the same algorithm on multi-core model with slight variations to bound block misses. 
Given that matrix-es are saved in row major layout, the algorithm can be described with this simple recurrence relation :
\begin{equation*}
\mbox{$T(m,n)$ = }
\begin{cases}
O(1) &\text{if $Max(m,n)\leq$ $c$}\\
O(1) + 2T(m,n/2) & \text{if n $\geq$ m}\\
O(1) + 2T(m/2,n) & \text{otherwise}
\end{cases}
\end{equation*}
where $c$ is some constant.
\begin{enumerate}
\item If both $m$ and $n$ are constant, transpose it directly using some brute-force approach.
\item Else, divide the whole matrix into two equal halves according to larger dimension.
\end{enumerate}
\textbf{Implementing with $p$ processors :} Say $K=mn/p$ and as given $K\geq M$. \\
\begin{equation*}
\mbox{$T^{\prime \prime}(m,n,p)$ = }
\begin{cases}
O(mn) &\text{if $mn\leq K$}\\
O(1) + T^{\prime \prime}(m,n/2,p/2) & \text{if n $\geq$ m}\\
O(1) + T^{\prime \prime}(m/2,n,p/2) & \text{otherwise}
\end{cases}
\end{equation*}
which gives us O$(mn/p+ \log p)$ parallel time.\\
After every level, both, size of problem and one of the dimension (row or column) is divided by two. As we keep on dividing the problem, at some level,
$mn$ becomes $\leq M$ and whole problem can fits into cache. Because of dividing nature of this algorithm, at that time one of 
the dimension will be $> \sqrt{M}/2$ and other will be $\geq \sqrt{M}$.\\   
\textbf{Processor Allocation} can be done using same method as explained in Prefix Sum Computation (Lemma \ref{prefix}) without any extra cost.\\
\textbf{Cache Misses :}
\begin{equation*}
\mbox{ $Q(m,n,p)$ = }
\begin{cases}
O(mn/B) &\text{if $mn\leq K$}\\
O(1) + 2Q(m,n/2,p/2) & \text{if n $>$ $m/4$}\\
O(1) + 2Q(m/2,n,p/2) & \text{otherwise}
\end{cases}
\end{equation*}
which comes to $Q(m,n,p)= O(mn/B).$ We have changed condition from $n\geq m$ to $n > m/4$ to make sure that when every processor gets its
individual problem, no. of rows are $\geq$ no. of columns in that sub-problem. This will help us to bound block misses.\\
\textbf{Block Misses :} After dividing given problem into enough sub-problems such that every core get its individual sub-problem, say of size $x\times y$, 
where we know that 
$xy  = mn/p \geq M$.\\
Depending on $m$ we have two cases:
\begin{enumerate}
\item  [Case I] $m\geq B$\\
Because of dividing conditions of this algorithm, we can claim that $x \geq \sqrt{M}$. While writing transpose of matrix $x \times y$ in output, 
processor will be writing $x$ contiguous elements $y$ times which will incur zero block misses, if $x\geq B$. 
\item  [Case II] $m < B$\\
\begin{figure}[ht]
\begin{center}
\includegraphics[scale=0.60]{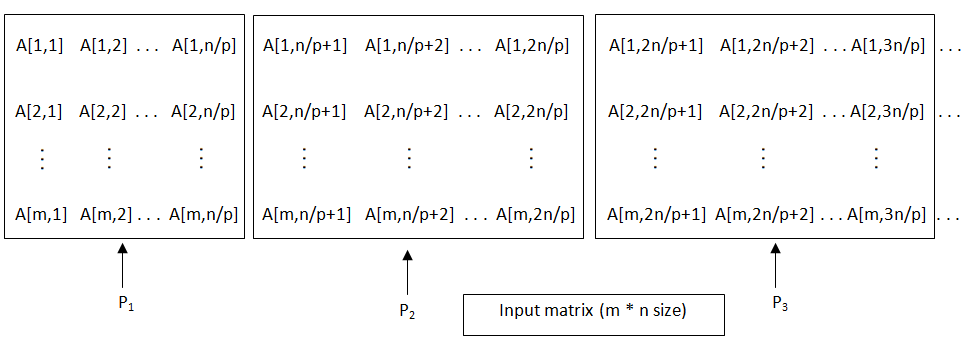}
\end{center} 
\caption{Each core gets matrix of size $m \times n/p$}
\label{imgtranspose1}
\end{figure}
\noindent{}When dividing given problem into sub-problems, we will always divide the matrix by column side till every processor get one 
sub-problem. Every processor will have a sub-problem of size $m \times n/p$ which is $\geq M$ as shown in figure $\ref{imgtranspose1}$.
Given that these matrix are saved in row major 
form, it can be easily seen that keys given to every core
will be stored in contiguous locations in output matrix as shown in figure $\ref{imgtranspose2}$.
\begin{figure}[ht]
\begin{center}
\includegraphics[scale=0.60]{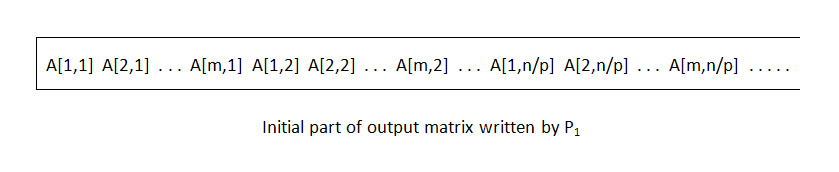}
\end{center} 
\caption{Each core gets write $mn/p$ contiguous elements}
\label{imgtranspose2}
\end{figure}
And every processor have $\geq M$ keys, so we will not have any block misses.
\end{enumerate}
This concludes the proof of theorem $\ref{transpose}$.  
\end{proof}
\begin{lem}\label{rank}
Given $n$ keys, rank of any key can be found in time
$O\bigg(\frac{n}{p}+\log p\bigg)$
with total $O(n/B)$ cache misses using $p$ cores, given $n\geq Bp$.
\end{lem}
\begin{proof} The whole task is divided into $p$ processors such that 
every processor will find rank of given key among $\lceil n/p\rceil $ keys.
It will take $\lceil n/p\rceil$ time and total $\bigg(\frac{n}{pB} \times p +p \bigg) = (n/B +p)$ cache misses
across all $p$ cores.\\
Ranks found by all $p$ cores are added to get the final rank of key among all $n$ keys.
 For addition, we will use only $p^2/n$ cores(to avoid block misses).
As $p\leq n$ is given, so $p^2/n$ will always be $\leq p$.\\
Using $p^2/n$ cores, $p$ elements can be added in
$O(n/p+\log p)$ time and $O(p/B)$ cache misses, given $p\geq B p^2/n$ or $n\geq Bp$ (Theorem $\ref{sum}$.) Total cache misses are bounded by $O(n/B)$, as $p\leq n$.
\end{proof}
\begin{lem}\label{compaction}
Using $p$ cores, $n$ elements can be written in contiguous locations in 
$O(n/p)$ time and $O(n)$ cache misses, if $n\geq Bp.$
\end{lem}
\begin{proof} The crux of this step is to avoid block misses. For that 
reason we make sure that every processor gets more than B contiguous locations to write
its keys.
This lemma is generally used in cases when processors work on disjoint problems, get some output 
and finally write all those outputs at contiguous locations.\\
Every processor is assigned an equal number of keys to write. 
The first processor reads the first $\lceil n/p\rceil $ keys from the input and 
writes them in the first $n/p$ locations
of the output. 
Similarly the  $i^{th}$
processor reads the $i^{th}$ block of $n/p$ keys from the input and writes
 in contiguous locations starting from $(n/p \times (i-1)+1)^{th}$ location of the output.
It will take $\lceil n/p\rceil $ time and $n$ cache misses to read and write all keys. Total block misses are zero if $n/p\geq B$.
\end{proof}
\subsection{Brute-force Sorting}
\begin{lem}\label{sortlog2}
Using Brute-force approach, $n$ keys can be sorted in $O\bigg(\frac{n^2}{p}+\log p\bigg)$  time using $p$ cores
with total cache misses $O(n^2/B+n)$, given $n^2\geq Bp$. 
\end{lem}
\begin{proof} To sort these $n$ keys, rank of every key is found in parallel using $(p/n)$ processors.\\
It will take $O\bigg(\frac{n^2}{p}+\log p\bigg)$ time and total $n \times \bigg(\frac{n}{B} +\frac{p}{n} \bigg) =
  \bigg(\frac{n^2}{B} +p \bigg)$ cache misses across all $p$ cores (Lemma $\ref{rank}$).
We have one output array B of length $n$ where every input key will be written at position given by its rank.
 To avoid block misses, use the following strategy.\\
Create one more output array A of length $n^2$ and write the 
 $i^{th}$ key at the location $ n \cdot R_i$ of $A$ (where $R_i$ represents rank of $i^{th}$ key.)
This is done in $\lceil n/p\rceil $ phases, by writing $p$ keys in every phase. In 
the first phase, all those cores whose keys have ranks $cn/p$
(where $c$ is $0$ to $n-1$), will write their keys and rest of the cores will remain idle. It will
take $O(1)$ time, $p$ cache misses and zero block misses if $n^2\geq Bp.$
In the next phase, all those cores whose keys have 
ranks $cn/p+1$, (where $c$ is $0$ to $n-1$) will write their keys.
In total, we will have total of $n/p$ such phases and
a total of $O(n)$ cache misses across all phases.\\
Finally, array A will contain all $n$ input keys spread-ed in $n^2$ locations in sorted order.
Using Lemma \ref{compaction}, we can write these $n$ keys at continuous location using $\lceil p/n\rceil $ processors in $O(n^2/p)$ time and $O(n)$ cache misses, if
$n\geq Bp/n$ or $n^2\geq Bp. $
\end{proof}
\begin{cor}\label{sortlog}
Using brute-force approach, $n$ keys can be sorted in time $O\bigg(\frac{n^2}{p}+\log p\bigg)$
with total $O(n^2/B)$ cache misses using $p$ cores, given $n^2\geq Bp$ and $n\geq B$.
\end{cor}
\begin{proof} This is a simple instance of lemma $\ref{sortlog2}$ when $n\geq B.$
\end{proof}
\subsection{Sampling}
\begin{lem}\label{sample}
Given a set $A$ of $n$ input keys and by using $p$ cores, where $n\geq Mp$, a set 
$S$ of size $n^{1/x}$ (where $x\geq 4$) can be chosen from A in $O(n/p+\log p)$ time and $O(n/B)$ cache misses such that 
$S$ satisfies the following property.\\ 
Suppose $t_0,\ldots, t_{f+1}$ is one of the longest sorted subsequence in sorted $A$ such that $t_0,t_{f+1}\in S$  
but $t_1,t_2,\ldots,t_{f}\notin $ $S$.\\
Then Pr(f $>$ $(1 + t^{-1/6}) n^{1-1/x}$  = $(t^{1/2}2^{(t^{1/2})})^{-1}$, where $t=\sqrt{n}/n^{1/x}$. 
\end{lem}
\begin{proof} This is done using the sampling technique given by Reif and Valiant\cite{refsample}.
\begin{enumerate}
\item Choose another set $S^{*}$ of size $n^{1/2}$  from $A$ randomly. The probability for any element being chosen in $S^{*}$ is equal. To choose 
$S^{*}$, imagine that input has been divided into
$n^{1/2}$ contiguous chunks all of equal size and one key is selected randomly from every chunk. It will take 
$(\lceil \sqrt{n}/p\rceil)$ parallel steps and total $n^{1/2}$ cache misses which is bounded by
$n/B$ if $n\geq B^2$ or $n\geq M$ and $m\geq B^2$. 
\item Using Lemma $\ref{sortlog}$, $S^{*}$ can be sorted in time O$\bigg(\frac{n}{p}+\log p\bigg)$ using $p$ cores
with total $O(\frac{n}{B})$ cache misses if $n\geq Bp$ and $\sqrt{n}\geq B$ which is true as $n\geq Mp$ is given.
\item To get required set $S$ of size $n^{1/x}$, choose $1(\sqrt{n}/n^{1/x})^{th}$,
$2(\sqrt{n}/n^{1/x})^{th}$, 
$ \ldots, \sqrt{n}^{th}$ element from $S^{*}$. It will take 
maximum $(\lceil n^{1/x}/p\rceil)$ parallel steps and total $n^{1/x}$ cache misses which is bounded by
$n/B$ if $n\geq B^2$.
\item Write these $\lceil n^{1/x}\rceil $ elements at contiguous locations using $\lceil pn^{1/x}/n\rceil $ cores in $\lceil n/p\rceil $ time and $\lceil n/B\rceil $ cache misses.
\end{enumerate}
From the results in \cite{refsample}, we can claim that 
Pr(f $>$ $(1 + t^{-1/6}) n^{1-1/x}$  = $(t^{1/2}2^{(t^{1/2})})^{-1}$, where $t=\sqrt{n}/n^{1/x}$.
\end{proof}

%% file: sorting.tex
\section{Randomized Divide and Conquer}

Our basic framework for randomized divide-and-conquer is exemplified
by Resichuk's \cite{refreischuk} parallel sorting algorithm. 
We are given $n$ keys in beginning $A=(k_1,k_2,\ldots,k_n)$. The basic steps 
of algorithm are :
\begin{enumerate}
\item {\bf (Random Sampling to divide)} Choose a subset $S$ of size $\sqrt{n}$ 
from $A$ randomly. 
\item Sort $S$ by comparing every pair of keys in $S$.
\item {\bf (Partitioning Step)} Using $S$, divide $A$ into $\sqrt{n}$ 
intervals $B_0,B_1,\ldots,B_{\sqrt{n}}$ where $B_i$ contains those keys that are bigger than
$i^{th}$ smallest but not bigger than $(i+1)^{th}$ smallest element of $S$.
\item {\bf (recursive call)} Sort each box recursively independent of other boxes.
\end{enumerate}
The above algorithm runs in $O(\log n)$ steps with high probability 
in an $n$ processors CRCW PRAM model \cite{refreischuk}. 
\subsection{Adaptation into Multicore}

We now present a variation of the above algorithm on Multi-core cache 
oblivious model to achieve the same kind of bounds as the PRAM model. 
We will use the following notations for the various performance measures:\\
$T$ : total time (total operations) \hspace*{0.2cm} $\bar{T}$ : expected time 
\hspace*{0.2cm} $T^{\prime \prime}$ : parallel time \hspace*{0.2cm}
$\bar{T}^{\prime \prime}$ : expected parallel time \\
$Q$ : total cache misses \hspace*{0.2cm} $\bar{Q}$: expected cache misses. 
\begin{thm}\label{reischuk}
Using $p$ cores, we can sort $n$ general keys in expected time $O(\frac{n}{p} \log + \log n \log\log n)$ 
and expected cache misses $O(\frac{n }{B}\log_M n)$, assuming $M\geq B^{2+\alpha}$ where $\alpha$ = 2/31 and $p\leq \frac{n}{M}$.
\end{thm} 
\begin{proof} We are given a set $A$ of $n$ keys $(k_1,k_2,\ldots,k_n)$ 
in shared memory. We have $p$ cores each with cache size of $M$.\\ 
Our algorithm divides the problem into disjoint subproblems recursively until
we have the number of subproblems matching the number of processors. 
Every subproblem is assigned processors according to its
size so that $(n/p)$ ratio is same for every sub-problem. As all the 
subproblems at same level have nearly equal sizes, so when problem size 
becomes lower than this ratio
then it will mean that we have enough number of subproblems.
(Note that we know $p$.)\\
We represent the initial input size as $N$ and total processors given as $P$. 
After every step we check whether problem size $\leq(N/P)$ if yes, 
each subproblem is assigned to an individual  core
otherwise partition it further into smaller subproblems.\\ 
\begin{figure}[ht]
\begin{center}
\includegraphics[scale=0.60]{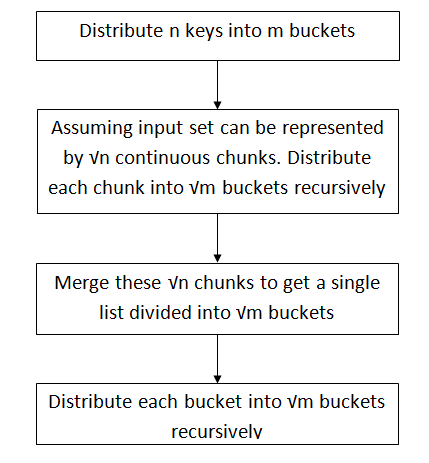}
\end{center}
\caption{The overall recursive partitioning scheme}
\label{blockdiag}
\end{figure}
The algorithms is executed in following steps :
\begin{enumerate}
\small{ 
\item If the problem size $\leq(N/P)$ then solve this problem on one 
processor sequentially. 
\\
\textit{Note that $n$ keys can be sorted using  one processor in time 
$O(n\log n)$ with $O(\frac{n }{B}\log_M n)$ cache misses (\cite{reftranspose})}
\item We choose a set $S$ of size $\lceil n^{1/x} \rceil,$ where $x\geq 4$  from $A$, using Lemma \ref{sample}. This task can be done in $O(n/p+\log p)$ time and $O(n/B)$ cache 
misses using $p$ cores, given $n\geq Mp.$ 
We will decide the value of $x$ later during the analysis.
\item Using $S$, divide $A$ into $|S|$ buckets $B_1,B_2,\ldots,B_{|S|}$ where 
$B_i$ contains those keys that are greater than
$i^{th}$ smallest but not greater than $(i+1)^{th}$ smallest element of $S$. \\
In other words, we have max \{x $\arrowvert$ x $\in B_i $\} $\leq S_i \leq$ min \{x $\arrowvert$ x $\in B_{i+1} $\}. \\
For this step we can follow the same approach as used in Reischuk sorting 
\cite{refreischuk}, where we assign every key to one processor and let it 
find the bucket number of this key using binary search on given $|S|$
bucket indexes. Then we do {\em integer sorting}\footnote{also known as
{\em semi-sorting}} to move all the 
elements to the right buckets. But even binary search step could cost us 
$n\log |S|$ cache misses where our target is 
$\frac{n}{B}\log_M |S|$ cache misses.
To avoid this we use the following strategy \\
Using Theorem \ref{dividemain}, we can do this task using $p$ cores in time $O(\frac{n}{p} \log n + \log n \log\log n)$
with a total of $O(\frac{n }{B}\log_M n)$ cache misses, provided $|S|\leq\sqrt{n}$, $p\leq \frac{n}{M}$ and $n\geq B^{2/(1 - \log_n |S|)}p$.
As $|S| = n^{1/x}$, this condition can be reduced to $n\geq B^{2/(1-1/x)}p$ 
or $n\geq B^2B^{2/(x-1)}p$. 
\item The original problem has been divided into $\lceil n^{1/x} \rceil$ subproblems where 
each subproblem can be solved independently. 
We claim that the size of largest subproblem will be  
$\leq (1 + t^{-1/6}) n^{1-1/x}$ with very high probability 
(which is $ 1 - 1/(t^{1/2}2^{(t^{1/2})})$, using Lemma  \ref{sample},
where $t=\sqrt{n}/n^{1/x}$. Therefore, a subproblem size is 
$\leq (1 + n^{(2-x)/12x}) n^{1-1/x}$ w.h.p. otherwise,
we repeat the partitioning step again starting from step two.
Distribute processors 
between subproblems according to their sizes and these subproblems are solved 
again recursively using this same procedure.
For the $n^{1/x}$ subproblems, processors 
are assigned to a problem such that the $(n/p)$ ratio
is same. 
}
For example, if the size of the $i^{th}$ problem is $n_i$ then 
it gets 
$p_i$ = $pn_i/n$ processors.
We will do prefix computation for the processor allocation. 
This takes $O\bigg(\frac{n}{p}+\log p\bigg)$ time 
using $p\frac{n^{1/x}}{n}$ cores with total cache misses $O(n^{1/x}/B)$, 
given $n^{1/x}\geq M\bigg(p\frac{n^{1/x}}{n}\bigg)$ or $n\geq Mp$. This concludes step 4. 
\\
Recall that to achieve the claimed bounds in steps 2 to 4, we require 
$\frac{n}{p}\geq$ max\{$M,B^2B^{2/(x-1)}$\}.\\
\subsection{Detailed Analysis}

We represent the initial input size as $N$ and the total number of 
processors given as $P$. As ratio $n/p$ remains same throughout the 
algorithm, so if it is given 
that $\frac{N}{P}\geq max\{M,B^2B^{2/(x-1)}\}$, condition given above will 
always be satisfied till the problem size $\geq (N/P)$.
When problem size becomes $N/P$, it means we have enough number of sub-problems so that we do not need to run steps 2 to 4, we solve the problem directly using step 1. \\
We can summarize the performance of our algorithm as : \\
Given, $\frac{N}{P}\geq max\{M,B^2B^{2/(x-1)}\}$, say $K=N/P,$ let 
$r$ be the probability of bad case (when size of largest sub-problem 
is $\geq$ $(1 + n^{(2-x)/12x}) n^{1-1/x}$), $x > 2$ and $P\leq N$. 
The expected parallel time satisfies 
\begin{equation}\label{neq1}
\bar{T^{\prime \prime}}(n)=
\begin{cases}
O\bigg(\frac{n}{p}\log n + \log n \log \log n\bigg) + (1-r) \bar{T^{\prime \prime}}(n_i) + r \bar{T^{\prime \prime}}(n) & \text{if $n \geq K$}\\
O(K \log K) & \text{otherwise}
\end{cases}
\end{equation}
For $n \geq 2^{18}$, we can choose $x$ s.t.
%where $n_i =(1 + n^{(2-x)/12x}) n^{1-1/x}$. We choose a value of $x$ such that 
%$(1 + n^{(2-x)/12x}) n^{1-1/x}\leq n^{31/32x}$ or 
$n_i\leq n^{1-1/2x}$. 
\ignore{
Choosing $x=32$, it can be reduced to 
$n^{31/(32 \times 32)}\geq(1+n^{-5/64})$ which is true if $n\geq 2^{16}2^{16/31}$, which is true as generally $B\geq 2^8$ and it is given that
$n\geq B^{2}B^{2/31}$. So by choosing $x = 32$ we can satisfy all conditions.\\
As $r < 1/2$ so $\frac{1}{(1-r)}<2$, after rearranging equation $(\ref{neq1})$, we get \\
} %ignore
Since $r \leq 1/n$, the previous recurrence can be rewritten as
\begin{equation}\label{neq2}
\bar{T^{\prime \prime}}(n)=
\begin{cases}
O(K \log n + \log n \log \log n) + \bar{T^{\prime \prime}}(n_i) & \text{if $n \geq K$}\\
O(K \log K) & \text{otherwise}
\end{cases}
\end{equation}
It follows that 
$\bar{T^{\prime \prime}}(n) = O\bigg(x\frac{n}{p}\log n+x\log n\log 
\log n\bigg) $ implying
$ \bar{T^{\prime \prime}}(n) = O\bigg(\frac{n}{p}\log n+\log n\log \log n\bigg).$
We have the following two cases depending on number of processors.
\\ Given $n\geq Mp$ and $n\geq B^2B^{2/(x-1)}p$ where $\bar{T^{\prime \prime}}(n)$ = $O(\frac{n}{p}\log n + \log n\log \log n)$. 
\begin{enumerate}
\item If $M \geq \log \log n$ or  $B^2B^{2/(x-1)} \geq \log \log n$ then $\frac{n}{p}\log n \geq \log n \log \log n$\\
so $\bar{T^{\prime \prime}}(n)=O(\frac{n}{p}\log n).$
\item If $p\leq \frac{n}{\log \log n}$ then $\frac{n}{p}\log n \geq \log n \log \log n$\\
so $\bar{T^{\prime \prime}}(n)=O(\frac{n}{p}\log n).$
\end{enumerate}
Likewise the total {\em expected} cache misses for all $p$ processors
\begin{equation}\label{neq3}
\bar{Q}(n)=
\begin{cases}
O(\frac{n}{B} \log_M n) + \displaystyle\sum_{i=1}^{n^{1/x}} \bar{Q}(n_i) & \text{if $n \geq K$}\\
O(\frac{K}{B} \log_M K) & \text{otherwise}
\end{cases}
\end{equation}
%\textbf{Cache misses :} \\ 
This yields $
\bar{Q}(n_i) = O(\frac{n_i}{B}\log_M n_i) + \displaystyle\sum_{i=1}^{n_i^{1/x}} \bar{Q}(n_{ii})$.\\
Since $n_i\leq n^{1-1/32x}$ and $\displaystyle\sum_{i=1}^{n^{1/x}} n_i=n$,
$\displaystyle\sum_{i=1}^{n^{1/x}} O(\frac{n_i}{B}\log_M n_i) \leq O(\frac{n}{B}\log_M n^{(1-1/32x)}) $.
%Similarly\\
%$ \bar{Q}(n)= a\bigg(\frac{n}{B}\log_M n\bigg) + a\bigg(\frac{n}{B}\log_M n^{(1-1/32x)}\bigg) + \ldots \ldots + 
%\displaystyle\sum_{i=1}^{n/K} b\frac{K}{B} \log_M K$\\ 
It follows that $\bar{Q}(n)$ = $O(x\frac{n}{B}\log_M n)$ implying
$\bar{Q}(n)$ = $O(\frac{n}{B}\log_M n).$\\\\ 
\ignore{
\textbf{Time :}
\begin{equation}
\bar{T^{\prime \prime}}(n)=
\begin{cases}
O(K\log n + \log n \log \log n) + \bar{T^{\prime \prime}}(n_i) & \text{if $n \geq K$}\\
O(K \log K) & \text{otherwise}
\end{cases}
\end{equation}
where $n_i\leq n^{(1-1/32x)}$, we get :\\
$
\bar{T^{\prime \prime}}(n)= a(K\log n + \log n \log \log n) + a(K\log n^{(1-1/32x)} + \log n^{(1-1/32x)} \log \log n^{(1-1/32x)})\\
$
$\mbox{\hspace{50pt}} + a(K\log n^{(1-1/32x)^2} + \log n^{(1-1/32x)^2} \log \log n^{(1-1/32x)^2}) + \ldots \ldots + bK\log K  \\
$
} %ignore
Notice that the bounds on time and caches misses are {\em expected} over
the choice of the random sample.
We have chosen $x = 32$, so if we assume that $M\geq B^2B^{2/31}$ then the 
condition $\frac{N}{P}\geq max\{M,B^2B^{2/(x-1)}\}$ can be simplified to 
$N\geq MP.$
This concludes the proof of Theorem \ref{reischuk}.    
\end{enumerate}
\end{proof}
\noindent{\bf Remark} : Using the analysis in \cite{refsample,3dhull}, we
can obtain high probability bounds for parallel time and cache misses.
\\
In the next subsections, we sketch some details of the individual steps.
\section{Merging algorithm}\label{mergesection}
\begin{figure}[ht]
\begin{center}
 \includegraphics[scale=0.50]{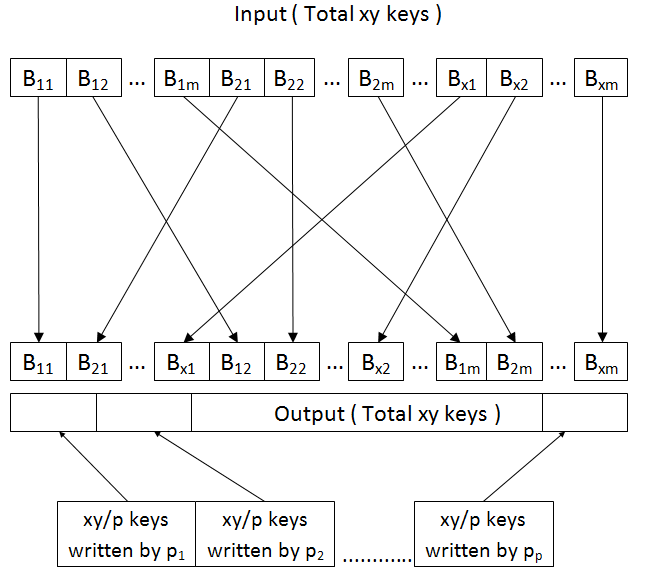}
\end{center}
\caption{Each core write contiguous $xy/p$ elements in output}
\label{imgmerge}
\end{figure}
\subsubsection{Basic structure of the Merging algorithm}
\textit{ Given $x$ lists each of size $y$ and all divided into $m$ buckets, merge all these lists to get one single list 
divided into $m$ buckets. }\\
Let $B_{ij}$ represent the $i^{th}$ list and $j^{th}$ bucket where 
$1 \leq i \leq x$ and $1 \leq j\leq m$. As shown in Figure \ref{imgmerge}, 
initially we are given input in form of $B_{11}$, $B_{12}$, \ldots, 
$B_{1m}$, $B_{21}$, $B_{22}$, \ldots, $B_{2m}$, \ldots, 
$B_{x1}$, $B_{x2}$, \ldots, $B_{xm}$ and
$p$ processors. We will divide the input into $p$ parts such that each core will get to write equal number of keys and in
a manner that in the output, the first processor will write the 
first $\frac{xy}{p}$ keys, the second processor will write next 
$\frac{xy}{p}$ keys and so on.
These processors may have to read their keys from different 
locations but while writing, they will write contiguously. 
To find the keys, every processor will need the information about the 
size of every bucket 
of every list. This will be done using prefix sum computation. \\
So whole task can be divided into three steps :
\begin{enumerate}
\item Every core needs to get the information about $xy/ p$ keys 
it must write. This is done using prefix sum computation. 
\item Every core will read exactly $xy/p$ keys using the information generated in first step.
\item Every core will write its $xy/ p$ keys in contiguous 
locations. 
\end{enumerate}
\subsubsection{Detailed explanation of Merging}
\begin{thm}\label{merge2}
Given $x$ lists of total size $y$, that  are indexed by $t$ buckets, 
we can merge these lists into a single list partitioned
into $t$ contiguous buckets in time $O\bigg(\frac{y}{p} +  \log p\bigg)$ 
using $p$ processors incurring $O(\frac{y}{B} + xt)$
cache misses, provided $y \geq Mp$ and $y\geq xt$.
\end{thm}
\begin{proof} Let $T_{merge}(x,y,t)$ 
represent the total time taken and $Q_{merge}(x,y,t)$ represent the total number of 
cache misses to merge $x$ lists of
total size $y$ and all divided into $t$ buckets using $p$ processors. \\ 
%These are the three steps explained above :
\begin{enumerate}
\item Every core needs to know all buckets size of every list to find $y/ p$ keys it has to 
write. There are a total of $xt$ such sizes. This information can be processed by first transposing given array of size $x\times t$
and then doing prefix sum computation on these $xt$ elements.\\ 
Using $(xpt/y)$ cores, both of these tasks can be done in time $O(y/p +  \log p)$ and $O(xt/B)$ 
cache misses, given $xt\geq M(xpt/y)$ or $y\geq Mp$ (Lemma $\ref{transpose}$ and $\ref{prefix}$).
We have enough cores for this task if $y\geq xt.$
\item From the information found in Step 1, we need those sizes which have values 
 $y/p, 2y/p, \ldots, py/p$. Every core is assigned a task to find these sizes 
from 
$xt/p$ sizes. Create an output array of length $y$, and if 
the size $cy/p$ is found, it will be written at location $cy/p$ of 
this output array. For reading and searching, it will take $\lceil xt/p\rceil $ time
and $\lceil \frac{xt}{B}+p \rceil $ cache misses. 
For writing, the total time will be $\lceil xt/p\rceil $, and the 
total cache misses are $p$ (and zero block 
misses if $y/p\geq B$).
\item These $p$ points found in Step 2 are written at contiguous locations in $\lceil y/p\rceil $ time 
with $O(p)$ cache misses, using $p^2/y$ cores, given 
$p\geq Bp^2/y$ or $y\geq Bp$ (Lemma \ref{compaction}).
\item Every core will read exactly ${y/p}$ keys which will take $\lceil y/p\rceil $ time. \\
Now we will bound total number of cache misses to read these $y/p$ keys. 
All the processors will read their input keys
sequentially unless two event happens as explained below :
\begin{enumerate}
\item The bucket that core was reading, ends. So in this case the processor 
may have to go to next list or next bucket which can increase cache miss 
count by at most
one over sequential misses and this can happen only $xt$ times. 
\item Processors may have to start or end reading in the middle
 of some block which can increase cache miss count by at most by one 
or two over sequential
misses. This can happen a maximum of $p$ times. 
\end{enumerate}
The total cache misses can be bounded by $O(y/B + xt + p)$ which is 
$O(y/B + xt)$ for $p\leq y/B.$.
\item Every core write its $y/ p$ keys sequentially that leads 
to a total of $O(y/B)$ cache misses. As $\frac{y}{B} \geq p$,
every core has more than one block to write and since 
every core is writing sequentially so there will not be any block misses. 
\end{enumerate}
This concludes the proof of Theorem \ref{merge2} .  
\end{proof}
The following is a simple corollary of the above Theorem.
\begin{cor}\label{merge1}
Given $x$ lists, each of size $y$, that are partitioned 
into $t$ buckets , we can merge these lists into a single list partitioned
into $t$ buckets in time $O\bigg(\frac{xy}{p} +  \log p\bigg)$ using $p$ processors with total $O(\frac{xy}{B} + xt)$
cache misses, given $xy \geq Mp$ and $y\geq t$. 
\end{cor}
\begin{cor}\label{merge3}
Given $x$ lists, of {\em total} size $Y$, we can merge these lists into a 
single list in time $O\bigg(\frac{Y}{p} + \log p\bigg)$ using $p$ processors
incurring a total of $O(\frac{Y}{B} + x)$ cache misses, given $Y \geq Mp$. 
\end{cor}
\begin{proof} By merging, we imply writing the input lists in contiguous 
locations without changing the order of elements, i.e.,
$L_1,L_2,L_3,\ldots,L_x$. This Lemma is used in cases when different 
processors work on different problems and generate some output, and
we need to write the
output in a sequence. So every list can be thought of as 
one bucket and the final list is also one list divided into one bucket.
This can be done by
using Theorem \ref{merge2}. Given $x$ lists of total 
size $Y$, we can merge these lists into a
 single list (divided
into one bucket) in time $O\bigg(\frac{Y}{p} +  \log p\bigg)$ using $p$ 
processors incurring a total of $O(\frac{Y}{B} + x)$
cache misses, given $Y \geq Mp$ and $Y\geq x$. Note that $Y$ is the 
total size of $x$ lists and every list has at least one element so 
$Y \geq x.$ 
\end{proof}
\section{Dividing input into buckets}
\subsubsection{Basic structure of algorithm}
Here we give a brief description about algorithm and later follow up with detailed time and cache analysis. \\
\textit{Given $n$ keys and $m$ bucket indexes  $X_1,X_2,\ldots,X_m$ (given in sorted order), divide $n$ keys into $m$ buckets $B_1,B_2,\ldots,B_{{m}}$
such that $X_i \leq $ \{x $\arrowvert$ x $\in B_{i} $\}  $\leq X_{i+1}$. \\\\ } 
\begin{figure}[ht]
\begin{center}
\includegraphics[scale=0.45]{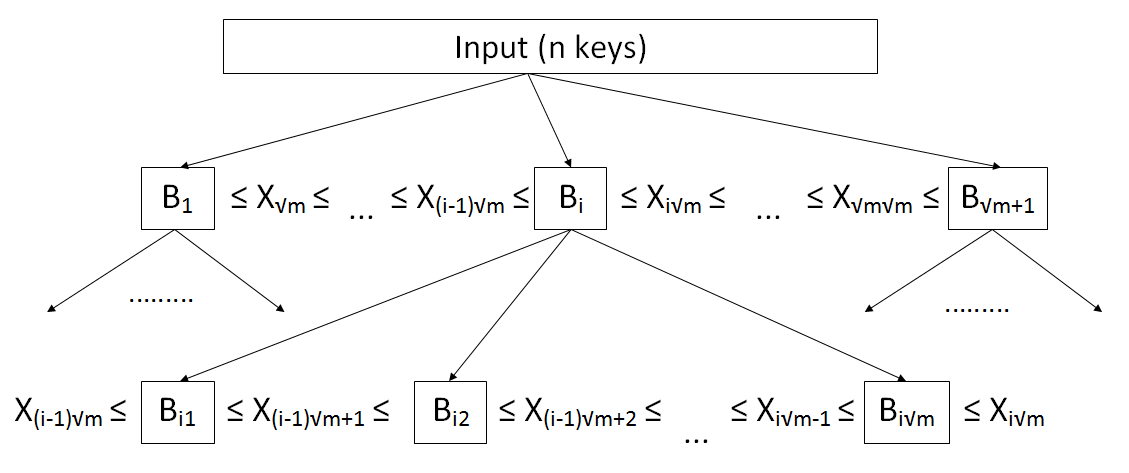}
\end{center} 
\caption{As shown for $B_i$, all buckets are divided into $\sqrt{m}$ buckets}
\label{imgdivide}
\end{figure}
Let  $T(n,m)$ represents the total time to divide $n$ keys into $m$ buckets.\\
First we divide $n$ keys into $\sqrt{m}$ buckets with splitters $( X_{\sqrt{m}},X_{2\sqrt{m}},\ldots,X_{\sqrt{m}\sqrt{m}} )$. For every $i$,
$1\leq i\leq \sqrt{m}$, $B_i$ is again divided into 
$B_{i1},B_{i2},\ldots,B_{i\sqrt{m}}$ as shown in figure \ref{imgdivide} :
\begin{enumerate}
\item Divide the $n$ keys into $\sqrt{m}$ buckets. This is done in two steps. 
\begin{enumerate}
\item Divide $n$ keys into $\sqrt{n}$ contiguous chunks of size $\sqrt{n}$ 
each. Now each chunk of size $\sqrt{n}$ is divided into
$\sqrt{m}$ buckets recursively. 
\item We get $\sqrt{n}$ lists, each divided into $\sqrt{m}$ buckets.
Merge these lists to get a single list divided into
$\sqrt{m}$ buckets. The merging algorithm is explained in the appendix. 
$$
  T(n,\sqrt{m}) = \sqrt{n}  T(\sqrt{n},\sqrt{m}) + T_{merge}(\sqrt{n},\sqrt{n},\sqrt{m})
$$
where $T_{merge}(\sqrt{n},\sqrt{n},\sqrt{m})$ represents time needed to merge $\sqrt{n}$ lists, each of size $\sqrt{n}$ and all divided into $\sqrt{m}$ buckets (explained in  section \ref{mergesection}) .
\end{enumerate}
\item  
Now we have $\sqrt{m}$ buckets $n_1$,$n_2$,$\ldots$,$n_{\sqrt{m}}$. Each bucket  is again divided into $\sqrt{m}$ buckets. 
$$
  T(n,m) = T(n,\sqrt{m}) +  \displaystyle\sum_{i=1}^{{m^{1/2}}} T(n_i,\sqrt{m})  
$$
We follow the same approach as done in step 1 above with some modifications. 
The following steps are done for every bucket :
\begin{enumerate}
\item Divide bucket into contiguous chunks of size $\sqrt{n}$. All these chunks(except last) will be of size
$\sqrt{n}$ and last chunk will have size $\leq \sqrt{n}. $
\item All these chunks (except the last) are divided into $\sqrt{m}$ buckets 
recursively. 
\item The last chunk is divided into $\sqrt{m}$ buckets directly with the 
help of binary search and sorting. We will explain this step in detail later. 
\item We have divided all these chunks in $\sqrt{m}$ buckets. Using our 
merge procedure, merge all these chunks to get one list divided
into $\sqrt{m}$ buckets. This concludes step 2.
\end{enumerate}
\end{enumerate}
Figure \ref{blockdiag} gives a high level structure of this scheme.
Finally, we obtain the following recurrence for the parallel time.
$$
  T^{\prime \prime}(n,m) = O(n/p +\log p) + 2T^{\prime \prime}(\sqrt{n},\sqrt{m}).   
$$
\subsubsection{Detailed analysis}
\begin{lem}\label{divide1}
Using one processor, $n$ keys can be divided into $x$ buckets in
$O(n\log n)$ time with $O(\frac{n }{B}\log_M n)$ cache misses, given $n\geq M$ and $x\leq n$ . 
\end{lem}
\begin{proof}
This problem can be easily reduced to sorting, so it will take the 
same time as that of 
sorting. We can sort $n$ keys in time  $O(n\log n)$ with 
$O(\frac{n }{B}\log_M n)$ cache misses using one processor 
(\cite{reftranspose}). After sorting, input is divided according to the 
bucket indexes. 
We just need to calculate the size of each bucket. Assuming bucket indexes are 
also given sorted. 
Start from the first bucket index and compare it with every key from the 
input and whenever it becomes larger than
this bucket index we switch to the next bucket index and start comparing it 
with next key from input. It will read both the 
input and the bucket
indexes contiguously from start. For $x\leq n$, the 
total cache misses for this step are bounded by $O(n/B)$. 
\end{proof}
\begin{lem}\label{divide2}
Using $p$ cores, $n$ keys can be divided into $x$ buckets in $O\bigg(\frac{n^{2}}{p} + \log p\bigg)$ time with $O(n^{2}/B)$ cache misses, 
given $n^2\geq Mp$ and $x \leq n$. 
\end{lem}
\begin{proof} We will solve this problem using the same approach as used in Lemma \ref{divide1}.\\
Sort $n$ keys in $O\bigg(\frac{n^2}{p} + \log p\bigg)$ time with total $O(n^2/B+n)$ cache misses,
given $n^2\geq Bp$ (Corollary $\ref{sortlog}$). The number of cache misses are bounded by $O(n^2/B)$ 
if $n\geq B$ (follows from $n\geq\sqrt{M}$ and $M\geq B^2$).
After sorting, input keys has been divided according to buckets.
To calculate the size of each bucket we assign every bucket index to one processor; say processor $p_i$ is assigned $i^{th}$ bucket index. 
$p_i$ does binary search to find the size of its bucket on $n$ keys. \\
To do this, $p$ processors will take $O(\frac{x\log n}{p})$ time and $O(xn/B)$ cache misses.
Given $x\leq n$, the time  and cache misses can be bounded by $O(\frac{n^2}{p})$ and $O(n^2/B)$ respectively.    
\end{proof}
\begin{lem}\label{divide3}
Using $p$ cores, $n$ keys can be divided into $y$ buckets in
$O\bigg(\frac{n^{3/2}}{p} + \log p\bigg)$ time with total $O(\frac{n^{3/2}}{B} + y\sqrt{n})$ cache misses, given $n\geq (Mp)^{2/3}$ and $y\leq \sqrt{n}$. 
\end{lem}
\begin{proof} Let $T(a,b)$ represent the total time to divide $a$ keys into $b$ buckets.
The whole algorithm can be described by the following recurrence relation.
$$
  T(n,y) = \sqrt{n}  T(\sqrt{n},y) + T_{merge}(\sqrt{n},\sqrt{n},y)
$$
where $T_{merge}(\sqrt{n},\sqrt{n},y)$ represents time needed to merge $\sqrt{n}$ lists, each of size $\sqrt{n}$ and divided into $y$ buckets 
(explained in  section \ref{mergesection}) .
\begin{enumerate}
\item Divide $n$ keys into $\sqrt{n}$ contiguous chunks, each of size $\sqrt{n}. $ All of these chunks are divided into
$\sqrt{y}$ buckets in parallel. 
Every $T(\sqrt{n},y)$ call is assigned to $p/\sqrt{n}$ processors and is solved using Lemma \ref{divide2}. It will take
$O\bigg(\frac{n^{3/2}}{p} + \log p\bigg)$ time and $O(n/B)$ cache misses if $n\geq M(p/\sqrt{n})$ or $n\geq (Mp)^{2/3}$ and $y\leq \sqrt{n}$.
Total cache misses for $\sqrt{n}$ calls will be $O(n^{3/2}/B)$ and the time is $O\bigg(\frac{n^{3/2}}{p} + \log p\bigg).$
\item From Step 1, we get $\sqrt{n}$ lists, each divided into $y$ buckets. Merge these $\sqrt{n}$ lists to get 
a single list divided into
$y$ buckets. Only $p/\sqrt{n}$ processors are used in this step to avoid block misses.\\
Using $p/\sqrt{n}$ processors, merging will take $O\bigg(\frac{n^{3/2}}{p} + \log p\bigg)$ 
time with $O(n/B + y\sqrt{n})$ cache misses, given $n\geq (Mp/\sqrt{n})$ or $n\geq (Mp)^{2/3}$ and $y\leq \sqrt{n}$ (Corollary $\ref{merge1}$). 
\end{enumerate}
\end{proof}
\begin{cor}\label{divide4}
Using  $(p/n^{1/4})$ cores, $n^{1/2}$ keys can be divided into $y$ buckets in
$O\bigg(\frac{n}{p} + \log p\bigg)$ time with total $O(n^{3/4}/B +y\sqrt[4]{n})$ cache misses, given $n\geq Mp$ and $y\leq \sqrt[4]{n}$. 
\end{cor}
\begin{proof}
This is an instance of Lemma \ref{divide3} with $\sqrt{n}$ input and $p/n^{1/4}$ cores
which lead to total $O(n^{3/4}/B + y\sqrt[4]{n})$ cache misses. To satisfy conditions given in Lemma \ref{divide3}, we need $\sqrt{n}\geq (Mp/n^{1/4})^{2/3} 
$ or $ n\geq Mp$ and $y\leq \sqrt[4]{n}.$
\end{proof}
\begin{thm}\label{dividemain}
We are given an array $A$ containing $n$ input keys and an array $S$ containing $z$ bucket indexes (stored in contiguous
locations). We can partition $A$ into $z$ buckets $B_1$ , $B_2$ , . . . , $B_z$ such 
that for
$i= 1,2$,. . . $z$, we have max \{x $\arrowvert$ x $\in B_i $\} $\leq S_i \leq$ min \{x $\arrowvert$ x $\in B_{i+1} $\}, using $p$ processors
in time $O(\frac{n}{p} \log n + \log n \log\log n)$ and a total 
of $O(\frac{n }{B}\log_M n)$ cache misses, given $z\leq \sqrt{n}$ and
$n\geq$max$\{Mp,B^{2/(1 - \log_n z)}p\}$.
\end{thm}
\begin{proof} 
Let $T(a,b)$ represent the total time, $T^{\prime \prime}(a,b)$ represent the
 parallel time and $Q(a,b)$ represent the total cache misses to 
divide $a$ keys into $b$ buckets. \\
All the processors work together to divide the problem into smaller size but 
when no. of sub-problems becomes more than the no. of processors then every core
can have one problem and solve it independently using Lemma \ref{divide1}. During recursion, all the sub-problems at the same level have the 
same size. Every sub-problem is
assigned processors according to its size so that $(n/p)$ ratio is same for every sub-problem at any level. When problem size becomes lower than this ratio
then it will mean that we have enough number of sub-problems to solve them on one processor individually.\\
Let initial input size be represented by $N$ and the 
total number of processors by P. It is assumed that $N\geq MP$.
The whole task can be divided into these steps. 
\begin{enumerate}
\item If problem size $\leq(N/P)$ then solve this problem on one processor otherwise break this problem into smaller sub-problems using the following steps.\\ 
\textit{Using Lemma \ref{divide1}, $n$ keys can be divided into $z$ buckets using  one processor in time 
$O(n\log n)$ with $O(\frac{n }{B}\log_M n)$ cache misses, given $z\leq n.$} 
\item Choose $\sqrt{z}^{th}$, $2\sqrt{z}^{th}$, $3\sqrt{z}^{th},\ldots,\sqrt{z} \sqrt{z}^{th}$ element from the given $z$
splitters. We obtain a total of $\sqrt{z}$ elements. First we divide input keys into these buckets and each bucket is divided further into $\sqrt{z}$ buckets.
For this step, it will take 
maximum $(\lceil z^{1/2}/p\rceil)$ parallel steps and one cache miss for every chosen key. Total cache misses will be $z^{1/2}$ and can be bounded by
$n/B$ as $z\leq \sqrt{n}$ and $n\geq B^2$.
\item Write these $\lceil z^{1/2}\rceil $ elements at contiguous locations using $\lceil pz^{1/2}/n\rceil $ cores in $\lceil n/p\rceil $ time,
 $\lceil n/B\rceil $ cache misses and zero block
misses, if $z^{1/2}\geq Bpz^{1/2}/n$ or $n\geq Bp$. 
\item Divide the $n$ keys into these $\sqrt{z}$ buckets as given as follows :
$$
  T^{\prime \prime}(n,\sqrt{z}) = T(\sqrt{n},\sqrt{z}) + T^{\prime \prime}_{merge}(\sqrt{n},\sqrt{n},\sqrt{z})
$$
$$
  Q(n,\sqrt{z}) = \sqrt{n}  Q(\sqrt{n},\sqrt{z}) + Q_{merge}(\sqrt{n},\sqrt{n},\sqrt{z})
$$
\begin{enumerate}
\item Divide whole input array into $\sqrt{n}$ contiguous sub-arrays of size $\sqrt{n}$. Recursively divide each sub-array of size $\sqrt{n}$ into
$\sqrt{z}$ buckets. 
\item  We get $\sqrt{n}$ lists, all divided into  $\sqrt{z}$ buckets. Using Corollary \ref{merge1}, merge these lists into a single list  
divided into  $\sqrt{z}$ buckets. Given, $n\geq$ Mp and $z\leq n.$
\end{enumerate}  
\begin{equation}\label{eq1}
\mbox{we get } T^{\prime \prime}(n,\sqrt{z}) =  T^{\prime \prime}(\sqrt{n},\sqrt{z}) + O\bigg(\frac{n}{p} +  \log p\bigg)
\end{equation}
\begin{equation*}
\mbox{Similarly for cache misses } Q(n,\sqrt{z}) = \sqrt{n}  Q(\sqrt{n},\sqrt{z}) + O(n/B +\sqrt{nz})
\end{equation*}
\item The entire input has been divided into $\sqrt{z}$ buckets and 
we need to divide each bucket again into $\sqrt{z}$ buckets. 
\begin{equation}\label{eq3}
  T(n,z) = T(n,\sqrt{z}) +  \displaystyle\sum_{i=1}^{{z^{1/2}}} T(n_i,\sqrt{z})  
\end{equation}
Break this $T(n_i,\sqrt{z})$ in chunks of $T(\sqrt{n},\sqrt{z})$ and then merge these chunks using Corollary \ref{merge1}. 
$n_i$ can be broken as  $n_i=x_i$ $\sqrt{n} + y_i$ where $y_i < \sqrt{n}$.
$\displaystyle\sum_{i=1}^{{z^{1/2}}} x_i \leq \sqrt{n} $    and    $ y_i < n^{1/2}$
$$
  T(n_i,\sqrt{z}) = x_i T(\sqrt{n},\sqrt{z}) +   T_{merge}(x_i,\sqrt{n},\sqrt{z}) +T(\sqrt{n},\sqrt{z})
$$
Summing this for all $z^{1/2}$ problems, we get
\begin{equation}\label{eq4}
\displaystyle\sum_{i=1}^{{z^{1/2}}} T(n_i,\sqrt{z}) = \sqrt{n} T(\sqrt{n},\sqrt{z}) +  \displaystyle\sum_{i=1}^{{z^{1/2}}} T_{merge}(x_i,\sqrt{n},\sqrt{z})
+ z^{1/2}T(n^{1/2},\sqrt{z})
\end{equation}
We have total $p$ processors and $z^{1/2}$ calls to $T(n^{1/2},\sqrt{z})$. Assign $p/n^{1/4}$ processors to each call and solve it using Corollary
\ref{divide4}.
As $z\leq \sqrt{n}$ so we have enough processors for $z^{1/2}$ calls. Each such call will take $O\bigg(\frac{n}{p} + \log p\bigg)$ time and $O(n^{3/4}/B+ \sqrt{z}\sqrt[4]{n})$ cache misses 
if $\sqrt{z}\leq \sqrt[4]{n}$ or $z\leq \sqrt{n}$ and $n\geq Mp$. For all
$n^{1/4}$ calls, total cache misses will be $O(n/B+\sqrt{nz})$.\\
Combining equation (\ref{eq1}),(\ref{eq3}) and (\ref{eq4}), we get
\begin{equation}\label{aeq1}
T^{\prime \prime}(n,z) = 2T^{\prime \prime}(\sqrt{n},\sqrt{z}) + O\bigg(\frac{n}{p} +  \log p\bigg) + T^{\prime \prime}_{merge}(x_i,\sqrt{n},\sqrt{z})
\end{equation}
Similarly for cache misses.
\begin{equation}\label{aeq2}
Q(n,z) = 2\sqrt{n}Q(\sqrt{n},\sqrt{z}) + O(n/B +\sqrt{nz}) +  \displaystyle\sum_{i=1}^{{z^{1/2}}} Q_{merge}(x_i,\sqrt{n},\sqrt{z})
\end{equation}
Assign $p_i$ processors to every bucket $n_i$ such that $p_i=pn_i/n$ or to be precise $p_i = \frac{px_i\sqrt{n}}{n} \Longrightarrow 
\frac{x_i\sqrt{n}}{p_i} = \frac{n}{p}.$
Using Corollary \ref{merge1}  
\begin{equation}\label{eq5}
\mbox{If $x_i\sqrt{n} \geq Mp_i$ then   } T^{\prime \prime}_{merge}(x_i,\sqrt{n},\sqrt{z}) =O\bigg(\frac{x_i \sqrt{n}}{p_i} +  \log p_i\bigg)
\end{equation}
and $x_i\sqrt{n} \geq Mp_i$ if $n\geq Mp$ and $z\leq n.$
As $p_i\leq p$, so equation (\ref{eq5}) can be reduced to
\begin{center}
$T^{\prime \prime}_{merge}(x_i,\sqrt{n},\sqrt{z}) $=\ $O\bigg(\frac{n}{p} +  \log p\bigg)$
\end{center}
For cache misses, $Q_{merge}(x_i,\sqrt{n},\sqrt{z}) $=$O(\frac{x_i \sqrt{n}}{B}+x_i\sqrt{z})$ (Corollary \ref{merge1}) \\\\
which gives $\displaystyle\sum_{i=1}^{{z^{1/2}}} Merge_Q(x_i,\sqrt{n},p_i) $ $\Longrightarrow$
$\displaystyle\sum_{i=1}^{{z^{1/2}}} O(\frac{x_i \sqrt{n}}{B}+x_i\sqrt{z})$\\$ \Longrightarrow O(\frac{\sqrt{n}}{B}+\sqrt{z}) 
\displaystyle\sum_{i=1}^{{z^{1/2}}} x_i $
$\Longrightarrow O(n/B+\sqrt{nz}).$\\\\
Finally equation (\ref{aeq1}) can be rewritten as 
$$
T^{\prime \prime}(n,z) = 2T^{\prime \prime}(\sqrt{n},\sqrt{z}) +   O\bigg(\frac{n}{p} +  \log p\bigg)
$$
Likewise for cache misses, using $n\geq Mp$, equation (\ref{aeq2}) can be reduced to
$$
Q(n,\sqrt{n}) = 2 \sqrt{n} Q(\sqrt{n},\sqrt{z}) + O(n/B+\sqrt{nz}) 
$$
The number of cache misses $O(n/B+\sqrt{nz})$ can be bounded 
by $O(n/B)$ if $\sqrt{nz}\leq n/B$ or $n\geq zB^2$.\\
Say $z=n^{1/t}\Rightarrow t=\log_z n.$
Condition $n\geq zB^2$ can be further reduced to $n^{1-1/t}\geq B^2 \Rightarrow n\geq B^{\frac{2}{1-1/t}} \Rightarrow n\geq B^{2/(1-\log_n z)} $\\
Finally, to complete all the steps of algorithm in claimed bounds, we need $n\geq$ max\{$Mp,B^q$\}, 
where $q=2/(1-\log_n z).$
As mentioned earlier, ratio $(n/p)$ remains same throughout the recursion 
procedure and the same is true for $q$ also. Since
after every iteration both $n$ and $z$ are reduced to $\sqrt{n}$, 
$\sqrt{z}$ and $\log_n z$ remain same, which means $q$ 
will also remain same through the recursion.\\
If it is given that $N\geq max\{MP,B^qP\}$, condition given above ($n\geq$ max\{$Mp,B^q$\}) is true till problem size $\geq (N/P)$
and when problem size becomes $N/P$, it means we have enough problems such that every processor can be assigned one problem.\\\\
\textbf{Processor Allocation :} During every iteration of this algorithm, one problem of size $n$ is divided into $\sqrt{n}$ sub-problems. Processors are assigned to 
every sub-problem such that $(n/p)$ ratio
is same for every sub-problem. The $i^{th}$ sub-problem with size $n_i$ gets $p_i$ processors 
such that $p_i=pn_i/n$, where $\displaystyle\sum_{i=1}^{n^{1/2}} n_i$ = $n$.
After doing prefix sum computation on sizes of all $n^{1/2}$ sub-problems, every processor can be assigned to its respective sub-problem. It will take $O\bigg(\frac{n}{p}+\log p\bigg)$ time 
using $p/\sqrt{n}$ cores with total cache misses $O(\sqrt{n}/B)$, given $n\geq Mp$.\\\\
The whole algorithm can be summarized as : \\
if $N\geq max\{MP,B^qP\}$, where $q=2/(1-\log_n z).$
\begin{equation}
T^{\prime \prime}(n,z)=
\begin{cases}
O\bigg(\frac{n}{p} +  \log p\bigg) + 2 T^{\prime \prime}(\sqrt{n},\sqrt{z}) & \text{if $n \geq N/P$}\\
O(\frac{N}{P} \log \frac{N}{P}) & \text{otherwise}
\end{cases}
\end{equation}
The total number of cache misses for all $p$ processors
\begin{equation}
Q(n,z)=
\begin{cases}
O(n/B) + 2 \sqrt{n} Q(\sqrt{n},\sqrt{z}) & \text{if $n \geq N/P$}\\
O(\frac{N}{BP} \log_M \frac{N}{P}) & \text{otherwise}
\end{cases}
\end{equation}
Assume $\frac{N}{P}=K$, where $K$ will be constant throughout the algorithm as explained before
\begin{equation*}
T^{\prime \prime}(n,z)=
\begin{cases}
a(K + \log n) + 2T^{\prime \prime}(\sqrt{n},\sqrt{z}) & \text{if $n > K$}\\
bK \log K & \mbox{otherwise}
\end{cases}
\end{equation*}
We get,
$T^{\prime \prime}(n,z)=a(K + \log n) + a(2K + \log n)$+ $\ldots $+ $a(2^xK + \log n) + b(2^{x+1}K\log K)$\\
$\Longrightarrow T^{\prime \prime}(n,z) =  2^{x+1}aK + ax\log n + 2^{x+1}bK \log K$\\\\
Using $(n)^{1/2^x} \geq K \Longrightarrow 2^x \leq \bigg(\frac{\log n}{\log K}\bigg)$\\
We get, 
$T^{\prime \prime}(n,z) = O(K\frac{\log n}{\log K} + \log \bigg(\frac{\log n}{\log K}\bigg)\log n + K\log n),$
which comes to\\\\ $T^{\prime \prime}(n,z)= O(K\log n + \log n \log \log n)$ Or $O(\frac{n}{p}\log n + \log n\log \log n)$. \\\\
The total cache misses for all the $p$  processors
\begin{equation}\label{eq8}
Q(n,z)=
\begin{cases}
an/B + 2 \sqrt{n} Q(\sqrt{n},\sqrt{z}) & \text{if $n > K$}\\
b\frac{K}{B} \log_M K. & \text{otherwise}
\end{cases}
\end{equation}
$Q(n,z) = a(\frac{n}{B}) + 2a(\frac{n}{B}) + 4a(\frac{n}{B}) + \ldots + 2^xa(\frac{n}{B}) + 2^{x+1}\frac{n}{K}a(\frac{k}{B} \log_M K)$\\\\
$\Longrightarrow Q(n,z) = 2^{x+1}a(\frac{n}{B}) + 2^{x+1}a(\frac{n}{B} \log_M K)$\\\\
Using $(n)^{1/2^x} \geq K \Longrightarrow 2^x \leq \bigg(\frac{\log n}{\log K}\bigg)$\\
We get, 
$Q(n,z) = O(\frac{n}{B} \log_K n + \frac{n}{B} \log_M n)$ = $O(\frac{n}{B}\log_M n) $ if $K \geq M$ or $n \geq Mp$.
\end{enumerate}
\end{proof}
\subsection{Multi-search}
\begin{lem}\label{multisearch}
We can simultaneously search for $n$ elements in a sorted set of 
$m$ elements using $p$ processors
in $O(\frac{n}{p} \log n + \log n \log\log n)$ time with total $O(\frac{n }{B}\log_M n)$ cache misses, given $m\leq \sqrt{n}$ and 
$n\geq$max$\{Mp,B^{2/(1 - \log_n m)}p\}$. 
\end{lem}
\begin{proof}
 This follows directly from Theorem \ref{dividemain}. The basic approach
goes back to Reif and Sen \cite{RS:94}.\\
 The condition $n\geq B^{2/(1 - \log_n m)}p$ can be simplified to 
$n\geq B^4p$ when $m = \sqrt{n}.$ 
As the value of $m$ decreases, this condition becomes progressively 
weaker. For $m=n^{1/32}$, condition is reduced to $n\geq B^{(2+2/31)}p$.
\end{proof}

%% file: convex.tex
\section{Convex Hull construction}
An object is \textit{convex} if for every pair of points within the object, 
the straight line segment that joins them lies completely within the object. \\
\textit{Given a set of points $P$, 
the {\em convex hull} $CH(P)$ is the smallest convex object which 
contains $P$.\\}
Because of close relationship between convex hull and sorting,
we will try to
prove the same time bound and cache cost for convex hull problem as our sorting algorithm. Our 
algorithm is based on the algorithm given in Reif and Sen \cite{3dhull} 
and our description
follows the randomized divide-and-conquer strategy described in \cite{ger}. 
It has been known that the convex hull problem can be reduced to the 
problem of finding the intersection of half-planes (under a dual transform).
So we will solve the problem of finding the intersection of half-planes as 
it is relatively easier to divide it into disjoint sub-problems.
\section{Basic Algorithms}
Before we describe the main algorithm, we briefly describe some of the
basic subroutines that are used frequently.
\subsection{2-D Maxima Problem} 
$A$ point $p = (p_x,p_y)$ dominates a point $q = (q_x,q_y)$ if and only if $p_x > q_x$ and $p_y > q_y$. The maxima problem
can be defined as : \textit{Given a set of points, to report the maximal points within the set} 
i.e. to report all the points which are not dominated by any other point in this set.\\
\begin{figure}[ht]
\begin{center}
\includegraphics[scale=0.60]{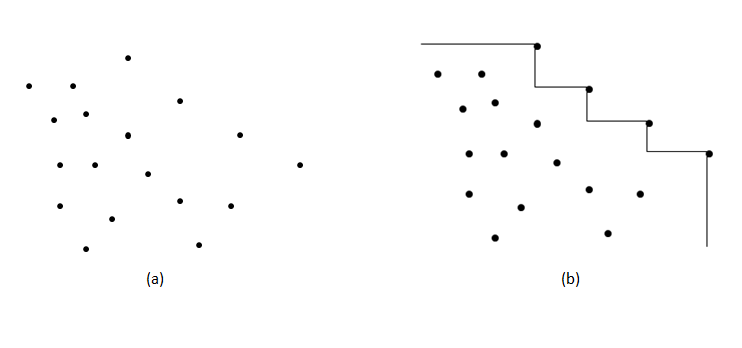}
\end{center} 
\caption{(a) Input points \hspace{35pt}  (b) Maximal points form the stairs}
\label{imgmax}
\end{figure}
This problem can be divided into following steps :
\begin{enumerate}
\item Sort the given points along x-axis to get a set $T$ = $\{t_1,t_{2},\ldots,t_n\}.$
\item It is obvious that $t_{n}$ is a maximal point. Start sweeping from $t_{n}$ towards left (in decreasing value of $x$).
The first point we encounter is $p_{n-1}$. We already know that 
$t_{n-1}(x)< t_n(x)$, so if $t_{n-1}(y)< t_n(y)$ then $t_{n-1}$ is dominated by $t_n$.
 Similarly, for the next point in left direction $t_{n-2}$, we just need to compare
$t_{n-2}(y)$ with max$\{t_n(y),t_{n-1}(y)\}$ to verify whether point $t_{n-2}$ is maximal point or not. \\  
In short, to check any point $t_i$ whether it is maximal or not, we just need to compare
$t_{i}(y)$ with max$\{t_{i+1}(y),t_{i+2}(y),\ldots,t_{n}(y)\}$. While sweeping to left, maintaining the value of the largest $y$-coordinate till 
that point is enough for comparison.
\end{enumerate}
\begin{lem}\label{cmaximal1}
From a given set S of $n$ points, maximal points can be found in $O(n\log n)$ time with total $O(\frac{n}{B}\log_M n)$ cache misses, using one
processor, given $n\geq M$.
\end{lem}
\begin{proof} The whole task can be divided into following two steps :
\begin{enumerate}
 \item Sort given $n$ keys w.r.t. x-axis in $O(n\log n)$ time with total $O(\frac{n }{B}\log_M n)$ cache misses using one processor, given $n\geq M.$ $[\cite{singlesort}]$.
\item As explained above, this sorted set is read sequentially which will cost $O(n/B)$ cache misses. 
We need to maintain one max variable which will be accessed every time next element is read from input, because of LRU replacement policy,
 it will not cost any extra cache miss.
\end{enumerate}
\end{proof}
\subsection{Finding maximal points on $p$ cores}
\begin{lem}\label{cmaximal2}
From a given set S of $n$ points, maximal points can be found in expected $O(\frac{n}{p} \log n + \log n \log\log n)$ time
with expected $O(\frac{n }{B}\log_M n)$ cache misses using $p$ cores, given $n\geq Mp$ and $M\geq B^2B^{2/31}$.
\end{lem}
\begin{proof}
The whole task can be divided into following steps :
\begin{enumerate}
 \item Sort input points w.r.t. $x$ co-ordinate in expected time $O(\frac{n}{p} \log n + \log n \log\log n)$ 
with expected $O(\frac{n }{B}\log_M n)$  cache misses using $p$ cores, given $n\geq Mp$ and $M\geq B^2B^{2/31}$ (Theorem $\ref{reischuk}$).
 \item Assume input has been divided into $p$ contiguous chunks $k_1,k_2,\ldots,k_p$, each of size $n/p$. Each core is assigned one chunk
to find maximum element say $y$, from it. To do so, every core will read these $n/p$ keys sequentially.
This will take $n/p$ time and total $n/B+p$ cache misses across all cores. We get total $p$ elements.
\item Write these $p$ elements at $p$ contiguous locations
using $p^2/n$ cores in $n/p$ time with total $p$ cache misses, if $p\geq Bp^2/n$ or $n\geq Bp$ (Lemma $\ref{compaction}$). 
Total cache misses will be $O(n/B)$ if $n\geq Bp.$
\item Do prefix computation using $max$ operator on these $p$ points found is Step 3. To avoid block misses use only $p^2/n$ processors.
 Using $p^2/n$ cores, prefix computation on $p$ keys can be done in O$(n/p+ \log p)$ time with total $O(p/B)$ cache misses given $p\geq Mp^2/n$ 
or $n\geq Mp$ (Lemma $\ref{prefix} $).
\item Every core is assigned a task to find maximal points from one chunk. It will take $O(\frac{n}{p}\log \frac{n}{p})$ time and total $O(\frac{n }{B}\log_M n)$
cache misses across all $p$ cores, if $n/p\geq M.$ (Lemma $\ref{cmaximal1}$).
\item Every core have one output list of maximal points from its corresponding chunk. Store these lists at contiguous locations. Using $qp/n$ cores,
we can join these $p$ lists of total size $q$,(where $q\leq n$) in time $O(n/p+\log p)$ time and $O(q/B)$ cache misses, if $q\geq Mqp/n$ or
 $n\geq Mp$ (Theorem $\ref{merge3}$).
\end{enumerate}
\end{proof}
\subsection{Dividing Convex hull problem into Upper hull $\&$ Lower hull}
\begin{lem}\label{cdivide2}
Given a problem to calculate convex hull of $n$ points, it can be divided into two disjoints problem of finding Upper hull and Lower hull of size $m$ and $n-m$ 
respectively 
in $O(n/p+\log p)$ time with total $O(n/B)$ cache misses using $p$ cores, given $n\geq Mp$.
\end{lem}
\begin{proof}
Divide the given problem into following simple steps :
\begin{enumerate}
 \item  Find the point with maximum $x$ co-ordinate and minimum $x$ co-ordinate [Lemma \ref{maximum}]. It will take $O(n/p+\log p)$ time and 
total $O(n/B)$ cache misses, if $n\geq Bp$. 
\item We get two points $p_1$ and $p_2$. For every input point check whether it is on upper side of line $p_1p_2$ or
lower side of line $p_1p_2$. Based on this information, input points can be divided into two disjoint problem. 
\item Every core reads $n/p$ keys and divide them into two buckets. This will take $\lceil n/p \rceil$ time and $n/B + p$ cache misses.
We will get $p$ lists each of each of $n/p$ size, all divided into two problems. 
\item Using lemma \ref{merge1}, merge these lists into a single list divided into two buckets
 in $O(n/p+\log p)$ time and total $O(n/B+2p)$ 
cache misses, given $n\geq Mp$ and $n/p\geq 2$ or $n\geq 2p.$ 
Total cache misses are bounded by $O(n/B)$ if $n\geq Mp.$
\end{enumerate}
Whole algorithm can be summarized by this recurrence relation.
$$
  T(n,2) = p T(n/p,2) + T_{merge}(p,n/p,2)
$$
\end{proof}
\subsection{Brute-force algorithm for intersection of half-planes}
\begin{cor}\label{crank}
Given two points $p_1$, $p_2$ and $n$ half-planes, find number of half-planes for which $p_1$ and $p_2$ are on same side in $O\bigg(\frac{n}{p}+\log p\bigg)$ time 
and total $O(n/B)$ cache misses, given $n\geq Bp$. 
\end{cor}
\begin{proof}
This is simple instance of Lemma \ref{rank} where rank of point can be defined as the number of half-planes for which both given points are on the same side. 
\end{proof}
\begin{lem}\label{cbrute}
Intersection of $n$ half-planes can be computed in $O(n^3/p+\log p)$  time using $p$ cores
with total $O(n^3/B)$ cache misses, given $n^3\geq Mp$ and a point $x$ which will be present in this intersection. 
\end{lem}
\begin{proof} We will use a simple brute-force approach. We have $n$ half-planes $k_1$,$k_2$,$\ldots$,$k_n$ and a point $x$ for which we already know 
that it will be inside intersection. As $n^3\geq M$ and $M\geq B^2$, so we can say that $n^2\geq B$.
\begin{enumerate}
 \item As $n$ lines can have maximum $n^2$ intersection points. We will simply check for each point whether it is a vertex of convex hull or not.
A point $p$ is a vertex of convex hull if $p$ and $x$ both are on same side of every half-plane
For every point $k_i$, where $1\leq i\leq n^2$, find whether it is vertex of convex hull or not, in parallel using $p/n^2$ cores.
Using lemma $\ref{crank}$, each point will take $O\bigg(\frac{n^3}{p}+\log p\bigg)$ time 
with total $O(n/B)$ cache misses, given $n\geq Bp/n^2$ or $n^3\geq Bp$. Total cache misses for $n^2$ points will be $O(n^3/B)$.
\item If point $k_i$ is vertex of convex hull, then one of the processors
which are working on $k_n^{th}$ point will write this point in output.
Maximum $n$ cores are ready to write one point each. We need to write these $n$ points in $n$ contiguous locations.
Create an output array $A$ of length $n^3$ and key $k_i$ will be written at position $n^2k_i$ in A.
This will take $O(n/p)$ time, $n$ cache misses and zero block misses if $n^2\geq B$. $n$ cache misses can be bounded by $O(n^3/B)$ if $n^2\geq B$. 
We have written these $n$ elements in $n^3$ locations.
\item Write these $n$ points at $n$ contiguous locations,
in time $n^3/p$ with total $n$ cache misses, using $p/n^2$ cores  if $n\geq Bp/n^2$ or $n^3\geq Bp$ (Lemma $\ref{compaction}$). 
Total cache misses are bounded by $O(n^3/B)$ if $n^2\geq B.$
\item All these written points are vertices's of convex hull but they may not be ordered.
Divide these $n$ points into two disjoint problems Upper hull and Lower hull using left-most and rightmost point (Lemma $\ref{cdivide2}$)
with $p/n^2$ cores
in time $O(n^3/p+\log p)$ with total $O(n/B)$ cache misses, given $n\geq Mp/n^2$ or $n^3\geq Mp$.
\item Using lemma $\ref{sortlog2}$, sort both of these disjoint problems.
We can sort $n$ points along $x$-axis using $n^2/p$ cores in $O(n^2/p+\log p)$ time and $O(n^2/B + n)$ cache misses which can be bounded by $O(n^3/B)$
if $n^2\geq B$.
\end{enumerate}
This concludes the lemma.
\end{proof}
\section{Randomized Algorithm For Convex Hull construction}\label{mainch}
The main result in \cite{3dhull} states that {\em Given  $n$ half-planes 
in $R^3$, and 
a point $t$ inside their intersection, we can find intersection of these 
half-planes in $O(\log n)$ expected time using $n$ 
CREW PRAM processors}. This algorithm achieves the optimal speed up. The basic steps of this algorithm are : 
\begin{enumerate}
\item Select $O(\log n)$ samples each of size $\lfloor n^{\epsilon}\rfloor$ from K randomly, where $0 < \epsilon < 1$. 
\item Select a \textit{good} sample using \textit{Polling}.
\item Find intersection of half-planes in selected sample using any brute force method to get a convex hull with points $P_1,P_2,\ldots,P_x$ 
(where $x\leq \epsilon$.)
\item  This convex polygon can be divided in $x$ triangles (or sectors) of the form $P_1tP_2, P_2tP_3,$ $\ldots, P_{x-1}tP_x, P_xtP_1$,
 where $t$ is given point which is present inside this hull.
\item The original problem can be divided into $x$ disjoint sub-problems, one sub-problem for each sector. For every sector find the intersection of half-planes which intersect this sector.
\item As one half-plane can intersect many sectors so total size of all sub-problems may be much larger than size of original problem.
But because of Polling we can claim that total size of all sub-problems will be $O(n)$, and size of maximum sub-problem will be
 $\leq n^{1-\epsilon}\log n$. 
\item  Use a \textit{filtering} scheme to reduce total size of all sub-problems from $O(n)$ to $2n$.
\item If size of sub-problem is very small then solve it directly using any brute-force algorithm otherwise solve it recursively using this same procedure.
\end{enumerate}
\subsection{Adaptation to Multi-Core Model}
We will use the algorithm explained above with slight modifications on Multi-core cache oblivious model to achieve the
same kind of bounds as the PRAM model. 
\begin{thm}\label{mainhull}
Using $p$ cores, the intersection of $n$ half-planes can be computed 
in expected time 
$O(\frac{n}{p} \log n + \log n \log\log n)$ incurring an 
expected $O(\frac{n }{B}\log_M n)$ cache misses, 
given $n\geq Mp$ and $M\geq B^2B^{2/7}.$
\end{thm}
\begin{proof}
The whole algorithm can be divided into following steps: 
\begin{enumerate}
\item From $n$ half-planes, using \textit{polling}\footnote{An efficient 
re-sampling technique that is described later}, choose a good sample of size $n^\epsilon$ using $p$ cores in expected time 
$O(\frac{n}{p} \log n + \log n \log\log n)$ with expected $O(\frac{n }{B}\log_M n)$ cache misses, 
given $\epsilon \leq 1/8$, $n\geq$ max$\{Mp,B^{2/(1 - 4 \epsilon)}p\}$ and $M\geq B^2B^{2/31}$ (Theorem \ref{polling}).
\item Compute intersection of $n^\epsilon$ half-planes in 
$O(n/p+\log p)$  time using $n^{3\epsilon} p/n$ cores
incurring a total of $O(n^{3\epsilon}/B)$ cache misses, 
given $n^{3\epsilon}\geq Mn^{3\epsilon} p/n$ or $n\geq Mp$ and a point $O$ which will be present 
in this intersection. (Lemma \ref{cbrute}). No. of processors needed for this are $\leq$ $p$ and total cache misses can be bounded by
$O(n/B)$  if $\epsilon\leq 1/3$.
\item We get a convex hull with points $P_1,P_2,\ldots,P_x$ 
(where $x\leq n^\epsilon$.) Divide this convex polygon into $x$ 
triangles (or sectors) of the form $P_1OP_2,P_2OP_3,\ldots,$ $P_{x-1}OP_x,$ $P_xOP_1$,
where $O$ is some given point which is present inside this hull.
\item Divide the original problem into $ x = O(n^\epsilon)$ 
disjoint sub-problems.\\
Using $p$ processors, divide $n$ half-planes into $n^\epsilon$ disjoint sub-problems 
in expected $O(\frac{n}{p} \log n + \log n \log\log n)$ time with expected $O(\frac{n }{B}\log_M n)$ cache misses , 
if $n\geq n^{\epsilon(x+2)}$, $x\geq 2$, $M\geq B^2B^{2/31}$ and $n\geq$max$\{Mp,B^{2+4/x}p\}$ Lemma (\ref{dividehull}).\\
Choose $x=(\frac{1}{2\epsilon}-2)$. Condition $x\geq 2$ is reduced to $\epsilon\leq 1/8$ and $n\geq B^{2+4/x}p$ is reduced to
$n\geq B^{2/(1-4\epsilon)}p$. 
\item Because of the way sampling was done in Step 1 (using \textit{polling}), we can claim that total size of all sub-problems will be $O(n)$ \cite{3dhull}.
Use a \textit{filtering} scheme to reduce total size of all sub-problems from $O(n)$ to $2n$.\\
Using $n^\epsilon$ sectors, $O(n)$ half-planes can be filtered in expected time $O(\frac{n}{p} \log n + \log n \log\log n)$
with a total of $O(\frac{n }{B}\log_M n)$ expected cache misses, given $n\geq Mp$ and $M\geq B^2B^{2/31}$ (Theorem \ref{filtering2}).
\item After filtering, total size of all sub-problems is $\leq$ $2n$ and all these sub-problems has been divided into $\epsilon$ sectors.
To maintain the same $n/p$ ratio throughout the recursion, increase the number of processors
to $2p$. As it has been proved (\cite{ger}) that at any level of recursion, maximum size of total sub-problems will be $\leq 2n$, $2p$ processors will be enough 
for this algorithm. No. of processors are assigned to each sub-problem according to its size. For e.g. problem with size $x$ will get
$xp/n$ processors. For division of processors, prefix computation is done on $n^\epsilon$ sizes.
Using $n^\epsilon p/n$ processors, it will take $O\bigg(\frac{n}{p}+\log p\bigg)$ time 
and $O(n^\epsilon/B)$  cache misses, given $n^\epsilon \geq Mn^\epsilon p/n$ 
or $n\geq Mp$ (Lemma~\ref{prefix}).
\item Maximum size of any sub-problem can be $2n^{(1-\epsilon)}\log n$. Those sub-problems which have size $\leq n/p$ get only one processor and can be solved sequentially. We do not need to divide these sub-problems any further.
Using one processor, intersection of $b$ half-planes can be computed in $O(b\log b)$ time and $O(\frac{b}{B}\log_M b)$ cache misses \cite{2dhull}.
All other sub-problems are solved recursively using this same procedure.
\end{enumerate}
As mentioned above, to complete step 1-7 in claimed bounds we need $\epsilon \leq 1/8$, $n\geq$ max$\{Mp,B^{2/(1 - 4 \epsilon)}p\}$ and $M\geq B^2B^{2/31}.$
By choosing $\epsilon=1/32$, all these can be reduced to $n\geq Mp$ and $M\geq B^2B^{2/7}.$\\
This algorithm follows same kind of recurrence as showed in our sorting algorithm (Theorem \ref{reischuk}).\\
\begin{equation}
\bar{T^{\prime \prime}}(n)=
\begin{cases}
O(K \log n + \log n \log \log n) + \bar{T^{\prime \prime}}(n_i) & \text{if $n \geq K$}\\
O(K \log K) & \text{otherwise}
\end{cases}
\end{equation}
where $K=n/p$ and $n_i \leq 2n^{1-\epsilon}\log n$ or $n_i \leq 2n^{31/32}\log n$ or $n_i\leq n^{63/64}$ if $2\log n\leq n^{1/64}$ 
which will be true for large $n$.\\
Likewise total cache misses for all p processors
\begin{equation}
\bar{Q}(n)=
\begin{cases}
O(\frac{n}{B} \log_M n) + \displaystyle\sum_{i=1}^{n^{1/x}} \bar{Q}(n_i) & \text{if $n \geq K$}\\
O(\frac{K}{B} \log_M K) & \text{otherwise}
\end{cases}
\end{equation}
As solved in Theorem \ref{reischuk}, we get total expected time $O(\frac{n}{p} \log n + \log n \log\log n)$ and total $O(\frac{n }{B}\log_M n)$ expected cache misses, 
\end{proof}
\begin{cor}\label{useit}
Given $n$ half-spaces in $R^3$ that contains the origin, 
we can find their intersection using $p$ cores in expected time 
$O(\frac{n}{p} \log n + \log n \log\log n)$ with expected cache misses $O(\frac{n }{B}\log_M n)$, 
given $n\geq Mp$ and $M\geq B^2B^{2/7}.$
\end{cor}
\begin{proof}
The above algorithm readily generalizes to 3 dimensions using the algorithm
of Reif and Sen\cite{3dhull} where we use a simpler {\em filtering} step
described in Amato, Goodrich and Ramos \cite{AGR}. 
\end{proof}
\subsection{Intersection of Half-planes and Convex Hull}
Given a convex polygon divided into sectors and number of input half-planes, for every half-plane identify the sectors of convex hull it intersects.
As shown in \cite{ger}, this problem can be reduced to a point location problem.
In dual space these half-planes are mapped to points and vertices's of convex hull are mapped to lines. 
These lines will intersect with each other to create some regions. It can be proved that all those points which belong to same region
in dual space, in primal space they are the half-planes which intersect same set of sectors in this convex hull.
As a part of preprocessing, for every region find the set of sectors it intersect.
After preprocessing, for any half-plane, do point location on these lines in
the dual space. (These lines are vertices's of convex hull in primal space.)
After finding region we can list all the sectors which this half-plane 
intersects. 
For point location problem we will use the algorithm of 
Dobkin and Lipton \cite{refpoint}.
After preprocessing the half-planes, we can perform the following steps:

\begin{enumerate}
\item For any query point we can find the unique cell in which it is 
present using two binary searches. 
\item For every region we can list the set of sectors in constant time
where listing involves the start sector and the end sector.
Because of convexity of the polygon, we can say that half-plane which pass through these start and end sectors will also pass through the sectors
between these start and end sector.
\end{enumerate}

\begin{lem}\label{preprocess} Using $p$ cores, $m$ half-planes can be 
preprocessed in $O(m^4/p+\log p)$ time with $O(m^4/B)$ cache misses, 
given $m^4\geq Bp$ so that point location among the half-planes 
can be done using two binary searches.
\end{lem}

\begin{proof}
The algorithm executes in the following steps :
\begin{enumerate}
 \item For $m$ half-planes we can have a maximum of $m^2$ intersection points.
Sort these $m^2$ points w.r.t. $x$ co-ordinate. Using $p$ cores, It will take $O\bigg(\frac{m^4}{p}+\log p\bigg)$ time 
and total $O(m^4/B+m)$ cache misses, given $m^4\geq Bp$ (Theorem \ref{sortlog2}).
\item From Step 1, we get $m^2$ intervals or say vertical slabs. 
In each slab there will be $n$ non-intersecting half-planes. We will
 order these half-planes according to $y$ coordinate. Allocate equal number of processors to each slab, so that every slab will get $p/m^2$ processors. 
\item For every vertical slab find intersection points of $m$ half-planes with both of its sides. This task can be done easily in $\lceil m^3/p\rceil $ time 
and $\lceil (m/B +p/m^2)\rceil $
cache misses by assigning $m^3/p$ half-planes to each processor for every slab. Total cache misses for all vertical slabs will be $(m^3/B+p)$
which is bounded by $O(m^4/B)$ if $m^4\geq Bp.$
\item All given $m$ half-planes has been divided into $m^3$ regions and as said earlier, all points which belongs to same region will intersect same set of sectors.
We just need to find the set of sectors, any region will intersect.
Take one point from every region, map it to
half-plane in primal space and test it against all sectors using some brute force algorithm
to get a set of sectors for every region but as this polygon is convex we will only save start and end sector.\\
This is done for all $m^3$ regions. Assign $m^3/p$ of these regions to every processor.
Every processor will search linearly through all sectors to find the set of sectors for every half-plane. It will take a total of $\lceil m^4/p\rceil $ time and
$\frac{m^3}{p} \times \frac{m}{B} \times p= \lceil m^4/B\rceil $ cache misses across all $p$ cores.    
\end{enumerate} 

\end{proof}

\begin{lem}\label{point} Given $n$ points and arrangement of $m$ half-planes, for every point we can identify the region it will lie in,
in $O(\frac{n}{p} \log n + \log n \log\log n)$ time and $O(\frac{n }{B}\log_M n)$ cache misses
using $p$ cores, given $n\geq$max$\{Mp, B^{2+4/x}p\}$,$n\geq m^{x+2}$ and $x\geq 2$. 
\end{lem}

\begin{proof} The entire procedure can be divided into the following steps :
 \begin{enumerate}
\item Given $m^4\geq Bp$, pre-process $m$ half-plane using Lemma \ref{preprocess} in $O(m^4/p+ \log p)$ time with total $O(m^4/B)$ cache misses.
 If $n\geq m^4$, time  and total cache misses are bounded by $O(n/p+ \log p)$ and $O(n/B)$ respectively.
After pre-processing, $m$ planes have been divided into $m^2$ vertical slabs and every vertical slab is further divided into $m$ horizontal slabs. 
Using this information, we will locate the region for every input point.
 \item Divide all the input points into these vertical slabs.\\
Using Theorem \ref{dividemain}, $n$ points can be divided into $m^2$
vertical slabs in $O(\frac{n}{p} \log n + \log n \log\log n)$ time and total $O(\frac{n }{B}\log_M n)$ cache misses
using $p$ cores, given $m^4\leq n$ and 
$n\geq$max$\{Mp,B^{2/(1 - 2\log_n m)}p\}$.
\item All $n$ points have been divided into $m^2$ sub-problems (vertical slabs) where 
the $i^{th}$ sub-problem contains $s_i$ points.
Divide every sub-problem again into $m$ sub-problems (horizontal slabs) using following steps :
\begin{enumerate}
\item First consider sub-problems with $s_i$ $\geq m^x.$ This $x$ will be chosen after analysis.
In worst case, all the sub-problems may have size $\geq m^x$. 
Assign $s_ip/n$ processors to $i^{th}$ problem according to size of this problem.
\item Processor allocation is done by doing prefix computation on size of these $m^2$ sizes of sub-problems.
Using $m^2p/n$ cores, prefix computation on $m^2$ keys can be done in O$(n/p+ \log p)$ time with total $O(m^2/B)$ cache misses, given
 $m^2\geq Mm^2p/n$ or $n\geq Mp$ (Theorem \ref{prefix}). We will have enough processors for this task if $m^2\leq n$.
\item Divide $s_i$ points into $m$
slabs in $O(\frac{n}{p} \log n + \log n \log\log n)$ time and total $O(\frac{s_i}{B}\log_M n)$ cache misses
using $s_ip/n$ cores, given $s_i\geq m^2$ and
$\frac{s_i}{s_ip/n} \geq $ max$\{M,B^{2/(1 - \log_{s_i} m)}\}$ or \\ $n/p \geq $max$\{M,B^{2/(1 - \log_{s_i} m)}\}$ (Theorem \ref{dividemain}).\\
In the worst case this step is done for all $m^2$ problems.
The total number of cache misses are $\displaystyle\sum_{i=1}^{{m^2}}\frac{s_i}{B}\log_M n = \frac{n}{B}\log_M n$ for all $m^2$ problems.
The condition $s_i\geq m^2$ is true if $x\geq 2$ because we are doing this 
step only for those problem who have $s_i\geq m^x$.  
For the next condition, we can see that $B^{2/(1 - \log_{s_i} m)}$ decreases
 as we increase $s_i$, so if condition $n/p \geq B^{2/(1 - \log_{s_i} m)}$ 
is true for minimum $s_i$ then it will
be true for all $s_i$. The minimum $s_i$ is $m^x$, which reduce the condition to $n/p \geq B^{2/(1 - \log_{(m^x)} m)}$ or $n/p \geq B^{2x/(x-1)}$, where $x\geq 2$.
\item Solve rest of the sub-problems with sizes
 $< m^x$. In the worst case there will be $m^2$ such problems and all of them 
will have size $m^x$.\\
Using Theorem \ref{dividemain}, we can divide $m^x$ ($x$ is a small constant) points into $m$
slabs in $O(\frac{n}{p} \log m + \log m \log\log m)$ time and total $O(\frac{m^x}{B}\log_M m)$ cache misses
using $m^xp/n$ cores, given $m^x\geq  m^2$ and $\frac{m^x}{m^xp/n} \geq B^{2/(1 - \log_{(m^x)} m)}$ or $n/p \geq $max$\{M, B^{2x/(x-1)}\}$.
We will need a total of $m^{x+2}p/n$ processors for $m^2$ problems. So we have enough processors for this task
 if $n\geq m^{x+2}$. 
The total cache misses for all $m^2$ problems will be $O(\frac{m^{x+2}}{B}\log_M m)$ which can be bounded by $O(\frac{n}{B}\log_M n)$ if $n\geq m^{x+2}.$
\end{enumerate}
\end{enumerate}
After combining all of the conditions mentioned in steps above,
we get $n\geq m^4$, $x\geq 2$, $n\geq m^{x+2}$, $n/p \geq $ max$\{M,B^{2x/(x-1)}\}$ 
and $n/p \geq $ max$\{M,B^{2/(1 - 2\log_n m)}\}$. As $x \geq 2$ so $n\geq m^{x+2}$ implies that $n\geq m^4$.\\
As $B^{2/(1 - \log_{n} m)}$ decreases as we increase $n$, so if condition $n/p \geq B^{2/(1 - \log_{n} m)}$ 
is true for minimum $n$ then it will
be true for all $n$. Minimum $n$ possible is $m^{x+2}$ so $n/p \geq B^{2/(1 - 2\log_{(m^{x+2})} m)}$ or $n/p \geq B^{2+4/x}$ is enough, where $x\geq 2$.
As $B^{2+4/x}$ $\geq B^{2x/(x-1)}$ so  $n/p \geq $ max$\{M,B^{2+4/x}\}$ is enough.

\end{proof}

\begin{cor}\label{findsector}
Given $n$ half-planes and a convex polygon divided into $m$ radial sectors, 
for every half-plane $h$, all the sectors intersected by $h$ can be identified 
in $O(\frac{n}{p} \log n + \log n \log\log n)$ time and $O(\frac{n }{B}\log_M n)$ cache misses
using $p$ cores, given $n\geq m^{x+2}$, $x\geq 2$ and 
$n\geq$ max$\{Mp,B^{2+4/x}p\}$. 
\end{cor}

\begin{proof}
We get this result simply by combining Lemma \ref{preprocess} and Lemma 
\ref{point}.
Map $n$ half-planes to $n$ points in dual space and $m$ vertices's of convex polygon to $m$ half-planes.\\
Using Lemma \ref{point},given $n$ points and $m$ half-planes, for every point we can identify the region it lies in,
 in $O(\frac{n}{p} \log n + \log n \log\log n)$ time and $O(\frac{n }{B}\log_M n)$ cache misses
using $p$ cores, given $n\geq m^{x+2}$, $x\geq 2$ and 
$n\geq$max$\{Mp,B^{2+4/x}p\}$. \\
Because of preprocessing done on $m$ half-planes in Lemma \ref{point}, any region can be directly mapped to its set of sectors.
\end{proof}

\begin{lem}\label{dividehull}
Given $n$ half planes, and a convex polygon divided between $m$ radial 
sectors, for each sector $C_i$ where $ i \leq m$, we want to identify
all the half-planes intersecting $C_i$. If the cumulative 
output size is $O(n)$, this task can be done using $p$ processors
in expected time $O(\frac{n}{p} \log n + \log n \log\log n)$ with expected $O(\frac{n }{B}\log_M n)$ cache misses, 
given $n\geq m^{x+2}$, $x\geq 2$, $M\geq B^2B^{2/31}$ and $n\geq$max$\{Mp,B^{2+4/x}p\}.$
\end{lem}

\begin{proof} We can divide this task into following steps :
\begin{enumerate}
 \item Using Lemma \ref{findsector}, for every half-plane identify all 
the sectors this half-plane will intersect
 in $O(\frac{n}{p} \log n + \log n \log\log n)$ time and $O(\frac{n }{B}\log_M n)$ cache misses
using $p$ cores, given $n\geq m^{x+2}$, $x\geq 2$ and 
$n\geq$ max$\{Mp,B^{2+4/x}p\}$. 
\item For every half-plane $K_i$ (where $1\leq i \leq n$) if it 
intersect with sectors $S_j$,$S_{j+1}$,$S_{j+2}$,\ldots,$S_k$ (where $j\leq k$ and both $j$ and $k$ are $\leq m$), it means there will be 
$t_i=(k-j+1)$ copies of this half-plane in output list. So we need to write  $t_i$ copies of
$K_i$ and assign them number $j,j+1,j+2,\ldots,k$ respectively. There will be total $\displaystyle\sum_{i=1}^{{z}}t_i$ elements in output which is 
bounded by $an$ (for some constant as given). Sort this output list w.r.t. this assigned number.\\
 This step is explained in detail in following steps :
\begin{enumerate}
\item Assign $n/p$ half-planes to every core so that every core will find $t_i$ 
for all of its $n/p$ half-planes
 and add them. 
To add $n/p$ elements, one processor will take $O(n/p)$ time and $O(n/pB)$ cache misses. The total cache misses across all cores will be $n/B.$
\item Every core will write its $n/p$ half-planes in output list as explained above.
 To find the location where every core needs to write we will do prefix computation on these $p$ sums found in step above.
  Using $p^2/n$ cores, prefix computation on $p$ keys can be done in O$(n/p+ \log p)$ time with total $O(p/B)$ cache misses given $p\geq Mp^2/n$ 
 or $n\geq Mp$ (Lemma \ref{prefix} ).
\item Sort these $an$ elements using Theorem \ref{reischuk} w.r.t. number assigned to them above to get final output half-planes divided 
them into $m$ sectors.  
It will take $O(\frac{an}{p} \log n + \log n \log\log n)$ expected time and $O(\frac{an }{B}\log_M n)$ expected cache misses, 
given $an\geq Mp$ and $M\geq B^2B^{2/31}$ (for some constant $a$).
\end{enumerate}
\end{enumerate}
\end{proof}

\subsection{Sampling and polling}
To bound the total size of the sub-problems we need to find a good sample. 
A good sample means when we divide our
$n$ half-planes into disjoints sub-problems using this sample total size of sub-problems will be $\leq cn$ and size of maximum sub-problem will be $\leq 2n^{1-\epsilon}\log n$,
where $\epsilon$ is number of sub-problems.
We will find this good sample using sampling and polling as shown in \cite{ger}. First we choose $\log n$ samples each of size $n^\epsilon$ and a 
random set of $n/\log^4 n$ 
input half-planes for all of these samples. Then for every sample we divide its input sample into disjoint problems using this sample. 
Then based on total problem sizes of
all these samples, we choose a 
good sample from those $\log n$ samples \cite{ger}.

\begin{lem}\label{samplingsingle}
Using one processor, $x$ elements can be sampled from $n$ elements in $O(x\log x+n)$ time with 
total $O(\frac{n}{B}+\frac{x }{B}\log x)$ cache misses, 
using $x$ independent trials such that the probability 
of an element being chosen is same in each trial.
\end{lem}
\begin{proof} We proceed as follows :
 \begin{enumerate}
 \item Processor will generate $x$ random numbers in range $[1,n]$ and save them sequentially in an array of length $x$. It will take $x$ time and incur 
$x/B$ cache misses.
 \item Sort these $x$ keys in time $O(x\log x)$ with 
$O(\frac{x }{B}\log x)$ cache misses using one processor.
\item Read these $x$ numbers and input keys both from start and select those number whose rank is present
in these $x$ numbers. It will take total $(x+n)$ time and $(x+n)/B$ cache misses.
\end{enumerate}

\end{proof}

\begin{thm}\label{polling}
Using polling, a good sample of size $n^\epsilon$ can be chosen from $n$ half-planes in expected
$O(\frac{n}{p} \log n + \log n \log\log n)$ time with expected $O(\frac{n }{B}\log_M n)$ cache misses, using $p$ cores, 
given $\epsilon \leq 1/8$ and $n\geq$max$\{Mp,B^{2/(1 - 4 \epsilon)}p\}$ and $M\geq B^2B^{2/31}$.
\end{thm}
\begin{proof} We can divide the whole task into following steps :
\begin{enumerate}
\item Using $p$ cores, choose $\log n$ samples each of size $n^\epsilon$ from given $n$ half-planes randomly, where $0< \epsilon <1$.
For every sample, the probability for choosing any element is equal.
Imagine that the input has been divided into
$n^\epsilon$ contiguous chunks all of equal size and choose one key from every chunk.\\
 To choose one sample using $p/\log n$ cores, it will take 
$\lceil \log n$. $n^\epsilon/p\rceil  $ parallel steps and one cache miss for every chosen key implying total $n^\epsilon$ cache misses.
For all $\log n$ samples, total cache misses will be $O(\log n$ $n^{\epsilon})$.  The time is bounded by $O(\frac{n}{p}\log n)$ if $\epsilon \leq 1$
and the cache misses are bounded by 
$\leq \frac{n}{B}\log_M n$ if $n^{1-\epsilon}\geq B\log M$ 
(which is true if $n\geq B^{2/(1-\epsilon)}$).
\item Find intersection of half-planes for every sample.\\
Using $pn^{3\epsilon} /n \log n$ cores, intersection of $n^{\epsilon}$ half-planes can be computed in $O(\frac{n}{p}\log n+\log p)$  time 
with total $O(n^{3\epsilon}/B)$ cache misses, if $n^{3\epsilon}\geq Mpn^{3\epsilon}/n \log n$ or $n \log n\geq Mp$ (Lemma \ref{cbrute}). \\
The total cache misses for $\log n$ samples will be $O(\frac{n^{3\epsilon}}{B}\log n)$ which is bounded by $O(\frac{n}{B}\log_M n)$ 
if $n^{3\epsilon}\log M\leq n$ or $\epsilon \leq 1/6.$
The total processors needed are $pn^{3\epsilon}/n$ which is $\leq p$ if $\epsilon \leq 1/3.$ 
\item We have to choose a random subset of $n/\log^4 n$ planes for 
every sample. If we use the same method as used in Step 1, it will lead to total $O(n/\log^3 n)$ cache misses.\\
 To reduce cache misses we use the following strategy :
\begin{enumerate}
 \item We have to choose total $n/\log^3 n$ elements. Every core is assigned $n/p$ half-planes and gets a task to choose $\frac{n}{p \log^3 n}$
 elements from these $n/p$ elements.
 \item Using Lemma \ref{samplingsingle}, Every core will sample $\frac{n}{p \log^3 n}$ elements from $n/p$ elements in $O(n/p)$ time 
 and $O(n/pB)$ cache misses. Total cache misses across all cores will be $O(n/B)$. 
 \item Every core assigns a random number in the 
range $1$ to $\log n$ to all of these $\frac{n}{p \log^3 n}$ half-planes.
 \item Divide these $n$ half-planes into $\log n$ problems w.r.t. 
this assigned number using Theorem \ref{dividemain}.
Using $p$ processors
it will take $O(\frac{n}{p} \log n + \log n \log\log n)$ time and total $O(\frac{n }{B}\log_M n)$ cache misses, given 
$n\geq$ max$\{Mp,B^{2/(1 - \log_n \log n)}p\}$ and $\log n\leq \sqrt{n}.$ 
\end{enumerate}

\item To find a good sample what we need is that 
every sample should divide its input sample into disjoint sub-problems 
and based on the total size
of those sub-problems, we can pick a good sample from these.
\item Using intersection found in step 1, every sample will divide its input sample (chosen in step 3) into disjoint problems.\\
Given $n/\log^4 n$ half-planes (chosen in step 3) and a convex polygon divided into $n^\epsilon$ sectors (found in step 1),
for every half-plane we can find all sectors this half-plane will intersect
 in $O(\frac{n}{p} \log n + \log n \log\log n)$ time and $O(\frac{n }{B\log^4 n}\log_M n)$ cache misses
using $p/\log^4 n$ cores, given $\frac{n}{\log^4 n}\geq n^{\epsilon(x+2)}$, $x\geq 2$ and 
$n\geq$ max$\{Mp,B^{2+4/x}p\}$ (Corollary \ref{findsector}.) \\
We get an output list of size $n/\log^4 n$. To find total size of sub-problems we need to just add these $n/\log^4 n$ elements using $p/\log^4 n$
processors. It takes time O$(n/p+ \log p)$ with a total of 
$O(\frac{n }{B\log^4 n})$ cache misses, given $n\geq Bp$.
As both tasks are done for $\log n$ samples, so the total cache misses 
across all cores will be $O(\frac{n}{B \log^3 n}\log_M n)$ which is $\leq O(n/B)$
\item For every sample we have found the total size of the 
sub-problems. After finding the maximum size we have enough information 
to select a good sample.
 Using $p\log n/n$ cores, maximum of $\log n$ keys can be found in time O$(n/p+ \log p)$ with total $O(\log n/B)$ cache misses given
 $\log n\geq Bp\log n/n$ or $n\geq Bp$.
\end{enumerate} 
All of the above conditions can be combine as $n\geq n^{\epsilon(x+2)}\log^4 n$, $x\geq 2$ and 
$n\geq$ max$\{Mp,B^{2+4/x}p\}$
 $M\geq B^2B^{2/31}$, $\epsilon \leq 1/6,$ $n\geq B^{2/(1-\epsilon)},$ $n\geq$ max$\{Mp,B^{2/(1 - \log_n \log n)}p\}$. \\
By choosing $x=(\frac{1}{2\epsilon}-2)$, we obtain
$n\geq B^{2/(1-4\epsilon)}$, $\epsilon \leq 1/8,$  $M\geq B^2B^{2/31}$

\end{proof}

\subsection{Filtering}\label{ll}
The total size of sub-problems after obtaining a good sample is $O(n)$. 
We do some further pre-processing to reduce this size to exact size of parent problem.
This step is a kind of post-processing step after sampling.  
\begin{figure}[ht]
\begin{center}
\includegraphics[scale=0.60]{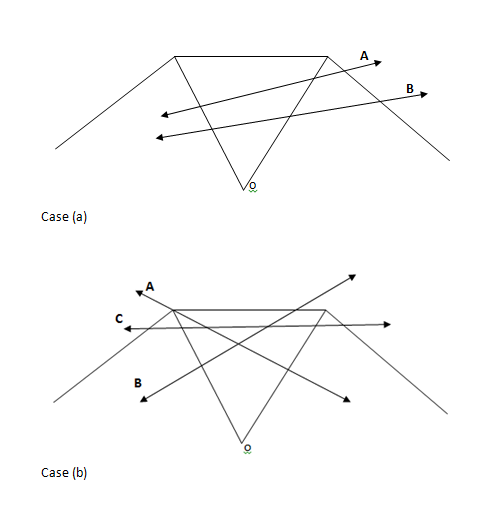}
\end{center} 
\caption{(a) $B$ dominates $A$\hspace{20pt}  (b) No half-plane is dominated}
\label{dominate1}
\end{figure}
For every sector we identify half-planes which intersect it. Some of them may be part of the
output while some of them may not show up in output (in that sector).
Since the output size is bounded by input size so our objective is to eliminate
all those half-planes from a sector which do not show up in output. This task can be reduced \cite{ger} to sorting and maximal points problem.\\
As shown in figure \ref{dominate1}(a) we can say that half-plane $A$ is dominated by $B$ through-out the sector
so it  would not be part of output in this sector. In filtering step we will remove all half-planes of these kind. After filtering there would not be any 
half-plane left in any sector such that it is dominated by any other half-plane as shown in figure \ref{dominate1}(b).\\
\textbf{ Filtering :} We are given some sectors and some half-planes. From these half-planes we will remove most of those planes which would not be part of 
output (intersection of given half-planes). If output size if $x$ 
(intersection of given half-planes),
then after this \textit{filtering} we will have utmost $2x$ half-planes left. 
\begin{thm}\label{filtering1}
Given a sector, $n$ half-planes can be filtered in expected time $O(\frac{n}{p} \log n + \log n \log\log n)$
with a total of $O(\frac{n }{B}\log_M n)$ expected cache misses , given $n\geq Mp$ and $M\geq B^2B^{2/31}$.
\end{thm}
\begin{proof} This task can be divided into the following steps :
\begin{enumerate}
 \item Consider one side of  a sector and find intersection points of half-planes with this side. This can be done in $\lceil n/p\rceil $ time and $\lceil n/B\rceil $
cache misses by assigning $n/p$ half-planes to each processor. 
\item Sort these half-planes w.r.t. distance of intersection point found above
 from center of sector. Using Theorem \ref{reischuk}
it will take expected time $O(\frac{n}{p} \log n + \log n \log\log n)$ and expected $O(\frac{n }{B}\log_M n)$ cache misses, 
given $n\geq Mp$ and $M\geq B^2B^{2/31}$.
\item Repeat steps 1 and 2 for other side of sector.
\item We get two sorted lists say $x$ and $y$. For half-plane 
$n_i$ consider a tuple of the form $x_i,y_i$, where $(x_i,y_i)$ are rank of these half-planes
in these sorted lists. To find this tuple we can do the following -  
In list $x$ write rank on every half-plane and then sort this list w.r.t. half-plane
number. We get a list of ranks for all half-planes.
\item We have a tuple $(x_i,y_i)$ for every half-plane. 
A half-plane $n_i$ is completely occluded by another half-plane $n_j$ if $x_j \geq x_i$
and $y_j \geq y_i$. In that case half-plane $n_i$ would not be part of output. Find maximal points on this tuple.
\item Using Theorem \ref{cmaximal2}, given  $n$  points, we can find the maximal points in expected time $O(\frac{n}{p} \log n + \log n \log\log n)$
with a total of $O(\frac{n }{B}\log_M n)$ expected cache misses using $p$ cores, given $n\geq Mp$ and $M\geq B^2B^{2/31}$. 
\end{enumerate}
\end{proof}

\begin{thm}\label{filtering2}
Given $n^\epsilon$ sectors, we can filter $O(n)$ half-planes in expected time $O(\frac{n}{p} \log n + \log n \log\log n)$
with a total of $O(\frac{n }{B}\log_M n)$ expected cache misses , given $n\geq Mp$ and $M\geq B^2B^{2/31}$.
\end{thm}

\begin{proof}
We have a total of $an$ half-planes divided into $n^\epsilon$ problems. 
Each problem has size $s_i$, (where $1\leq i \leq n^\epsilon$) and
$\displaystyle\sum_{i=1}^{{n^\epsilon}}s_i=an$.
We will do filtering on all of these $n^\epsilon$ problems in parallel. 
Every problem is assigned number of processors according to its size, 
so the $i^{th}$ 
problem is assigned $s_ip/an$ processors using prefix computation (Theorem \ref{prefix}).\\
Using Theorem $\ref{filtering1}$ described above, one problem of size $s_i$(where $s_i\leq an$) using $s_ip/an$ processors can be filtered in expected 
time $O(\frac{an}{p} \log n + \log n \log\log n)$
with a total of $O(\frac{s_i}{B}\log_M n)$ expected cache misses , given $s_i\geq Ms_ip/an$ or $an\geq Mp$ and $M\geq B^2B^{2/31}$.
The total cache misses for all the $n^\epsilon$ problems is 
$O(\frac{an}{B}\log_M n).$ \\
For any constant $a$, time is bounded by $O(\frac{n}{p} \log n + \log n \log\log n)$
and total number of cache misses is bounded by $O(\frac{n}{B}\log_M n).$  
\end{proof}

%% file: procalloc.tex
%\documentclass[11pt]{article}
%\begin{document}
\section{Concurrent Writes and processor oblivious load balancing}

Traditionally, parallel programs are written assuming that
there is a unique id (an integer in the range $1 \ldots p)$
for each of the $p$ processors. The processors id is used
to designate a particular task to a specific processor. For
example, for an input array of $n$ numbers, a processor with
id $i$ may be allocated the task associated with the sub-array
$\frac{n}{p} \cdot i \ldots \frac{n}{p} \cdot (i +1)$. This
is easy because $p$ is known at the time the parallel code is
generated. However, in some situations $p$ may not known and
moreover, the processors may not have an id associated with
them. 
\ignore{This variation is called the {\em resource oblivious}
case and it is becoming increasingly important to design parallel
algorithms where a processor's task is not defined on the basis
of its id. This requires us to design parallel scheduling algorithms
where the parallel tasks are allocated to a processor from a {\em pool}
on the fly as and when they are spawned. A commonly used strategy is 
called {\em work stealing} where processors maintain a queue of tasks
allocated to them and try to achieve load balancing by {\em stealing}
task from another processor if its own queue is empty \cite{cilk}.
In this case, the parallel algorithm and the processor scheduling behave
as independent distributed processes interacting on demand. This has
some clear advantages like inherent fault tolerance and self-adjusting
to a varying pool of available processors.

 In contrast to this, we can design parallel algorithms where the
number of processors is not known initially but the processors
can compute a unique id on-line and subsequently execute the traditional
parallel program based on this id. We shall describe a simple randomized
technique to achieve this that may have the advantage of being more
efficient than {\em work stealing} paradigm in practice.
}%ignore
We illustrate this method for the fundamental problem of prefix 
computation. The basic idea is that each of the $p$ processors simultaneously
chooses a random number in the range $[1..n]$ and writes to the corresponding
location in an array $A$. The expected number of processors writing to a
specific location $A[i] \ \ 1 \leq i \leq n$, is $p/n$ and no more than
$O(\log n)$ w.h.p. - note that, from our earlier assumptions, $p \leq n/B$,
so the expected number of elements writing into a $B$ element block $\leq 1$.   
Roughly speaking, a processor writing to a location $i$ assumes 
responsibility for the block of $n/p$ elements starting from 
$\lfloor \frac{i}{n/p}\rfloor$.
However, because of conflicts caused by independent random choices, we
have to do some limited redistribution and also {\em estimate} the value
of $p$. It can be argued using Chernoff bounds that w.h.p. $\Theta (\log n)$
processors will choose a location in the range $[bi \ldots b(i+1)] \ \ 
i = 0, 1 \ldots$ where $b = (n/p)\cdot \log n$.

From the previous observation, $p$ can be estimated in the following manner.
Let the {\em leftmost} processor in the the array
count up to $\log n$ processors. If this location is
$\beta$, then we can estimate $p$ to be within a constant factor of $(n/\beta)
\cdot \log n)$. This procedure requires more elaboration.

To find the leftmost processor, we can let every processor that succeeded
in writing to a location in the first step traverse left in a synchronous 
fashion. Once it encounters a location that has been already written into, it
terminates the traversal. Only the leftmost processor will be able to continue
and reach location 1.

In the second phase when it counts up to $\log n$ processors, we have to ensure
that the processors that conflict in the random choices are counted with the
proper multiplicity. If $j$ processors write to a single location, then we
let these processors increase a counter and the final count is the number
of the processors - since $j \leq \log n$ w.h.p., this step completes
in $O(\log n)$ time. So the leftmost processor can start its rightward 
traversal after $\Omega (\log n)$ steps and report the location $\beta$.
Note that the expected value of $\beta$ is $(n/p)\log n$ which can be subsumed
in the overall time for a problem like sorting but not for prefix computation.
So, for prefix computation, we can assume that $p \leq (n/\log n)$ and
run the estimation procedure in an array of length $n/\log n$, so that the
overall time for the estimation is $O((n/\log n) \cdot (1/p)\cdot \log n)$ which
is $O(n/p)$. 

The processors that belong to the block (of size $(n/p)\log n)$ must be assigned
a unique id which can be computed as follows. Once $p$ is estimated, each
processor knows its most significant $\log p - \log\log n$ bits of the id.
The remaining $\log\log n$ bits can be assigned on the basis of a simple
prefix sum within the block carried out by the leftmost processor
within the block (once $p$ is estimated, the leftmost processor in a
block is easily identified by repeating the procedure within the block).
Following this, each processor is assigned a unique block of $\theta(n/p)$ 
elements.

In the above procedure, we have not accounted for the {\em block misses} when
processors conflict in writing to a specific block. For instance, if $j$
processors write to the same block, the {\em block misses} will be
$\Omega ( j^2 )$. This could be as much as $\Omega (\log^2 n)$ since we
can only bound $j \leq \log n$ with high probability. The overall block misses
will be given by $\sum_{i=1}^{n'} n_i^2$ where $n' = n/B$ and
 $n_i$ is the number of
processors writing in block $i$. Note that $\sum_i n_i = p$. We can
compute the expected {\em block misses} as follows. Since $n_j$ denotes the
number of processors that chose block $j$, we are interested in the
quantity $E[\sum_j n_j^2]$ that represents the total expected block misses. 
Let r.v. $X_i = 1$ if processor $i$\footnote{this is only for the analysis 
since a processor is not explicitly given any id; also note that we need
Concurrent Write capability for this}  writes to block 1  and 
0 otherwise (the same
analysis will apply to any fixed block by symmetry). Then $E[X_i] = 1/n'$
and moreover $E [ X_i^2] = 1/n'$. From our earlier notation, $n_1 
= \sum_{i=1}^{p} X_i$ and so $n_1^2 = \sum_{i=1}^{p} X_i^2 $. Taking expectations,
\[ E[ {\left(\sum_{i=1}^{p} X_i \right)}^2 ] 
 = p E [ X_i^2 ] + {p \choose 2} \cdot 2 E [ X_i \cdot X_j ] 
 = p \cdot \frac{1}{n'} + p(p-1) \cdot E[X_i]\cdot E[X_j] 
 = p\cdot \frac{1}{n'} + p(p-1) \frac{1}{n'^2} \]
The first equality follows from linearity and the second from independence.
Therefore, from symmetry, $E[ n_1^2 + n_2^2 + \ldots ] = 
n' \cdot(p/n' + p(p-1)/n'^2)$ = $ p (1 + (p-1)/n')$ which is $O(p)$ for 
$p \leq n' = n/B$. 

 For prefix computation, we proceed as follows. 
Given $p \leq n/\log n$ processors, we first estimate $p$ using the
first $n/\log n$ locations (i.e. $n'/\log n$ blocks) 
of the array and also the processor id. If there
are $c \log n$ processors for a block of size $n/p \log n$, then a processor
with id = $<m', \ell'>$ where $m'$ denotes the most significant 
$\log p - \log\log n$ bits, and $\ell'$ denotes the least significant bits
is assigned the locations $[\alpha \cdot m' + \beta \cdot \ell' \ldots
\alpha \cdot m' +\beta (\ell' + 1)]$ where $\alpha = n/p \log n $ and
$\beta =  \frac{1}{c} \cdot (n/p)$.

The first phase takes $O(n/p)$ sequential steps, following which the number of
data items is $p$. We can run the 
optimal (cache oblivious) speed-up prefix computation on the $p$ processors.   
\ignore{
\subsection{Sorting random input}

Rajasekaran and Sen\cite{RS:90} described a simple optimal speed-up 
algorithm for
sorting $n$ integers in arbitrary range that runs in $O(\log n)$ expected
time in a CRCW model. The intuitive idea is as follows - hash the input
numbers into $n\log n$ buckets based on the key values. Because of the
property of uniform distribution, most buckets will contain less than $O(C)$
keys and moreover the total number of keys in the {\em heavy} buckets 
(containing more than $C$ keys) is $O(n/\log n)$ w.h.p. - denote this 
by $H$. We can sort $H$ using any general sorting algorithm and subsequently
merge with the remaining keys (they can be sorted using any trivial procedure).

Using the above analysis for bounding block misses, and using the 
previous algorithms for sorting and merging, we obtain the following
result.
\begin{theorem}
If $n$ keys are chosen uniformly at random, they can be sorted in
$O( n/p)$ time using $p$ CRCW multicore processors and expected
$O(n/B)$ cache misses, including block misses where $p \leq \min \{ 
n/B ,n/\log n \}$. 
\end{theorem}   
\noindent{\bf Proof}: Note that, no more than $\log n$
elements are written into the same block w.h.p. and in our model, it
will take at most $O(\log n)$ time following a sequential conflict
resolution. This also implies a simple mechanism for counting the
number of elements written in each location that can be used for finding
a unique location for each element in later stages of the algorithm.
Further details of the proof is omitted from this version.
}
\section{Future Work}
For both sorting and convex hull algorithm given here, running time is not optimal when we have $n$ cores.
The bottleneck step in these algorithm is where we divide input into buckets. 
This step is being implemented using a deterministic algorithm. We can try to optimize this by using a randomized approach.  

%% file: reference.tex
\begin{footnotesize}

\end{footnotesize}

%% file: paper.bbl
\begin{thebibliography}{99}
%\addcontentsline{toc}{chapter}{References}

\bibitem{2dhull} Peyman Afshani, Arash Farzan, Cache-Oblivious Output-Sensitive Two-Dimensional Convex Hull, 
19th Canadian Conference on Computational Geometry, 2007.
\bibitem{vitter} A. Aggarwal and J. S. Vitter. The Input/Output Complexity of Sorting and Related problems, 
Communications of the ACM, 1988.
\bibitem{AGR} N.M. Amato, M.T. Goodrich, and E.A. Ramos, 
N.M. Amato, M.T. Goodrich, and E.A. Ramos, Parallel algorithms for higher-
dimensional convex hulls, Proc. 35th Annu. IEEE Sympos. Found. Comput. Sci.,
pp. 683–694, 1994.
\bibitem{four} L. Arge, M. T. Goodrich, M. Nelson, and N. Sitchinava. Fundamental parallel
algorithms for private-cache chip multiprocessors, SPAA'08, June 14-16, 2008, Munich, Germany.
\bibitem{five}G. Blelloch, R. Chowdhury, P. Gibbons, V. Ramachandran, S. Chen, and
M. Kozuch. Provably good multicore cache performance for divide-and-conquer
algorithms. In ACM-SIAM SODA, pages 501-510, 2008.
\bibitem{cilk} Robert D. Blumofe and
               Charles E. Leiserson, Scheduling Multithreaded Computations
by Work Stealing, Journal of the ACM 46(5), 1999,pp. 720-748.
\bibitem{refprefix} Guy E. Blelloch, Phillip B. Gibbons, Harsha Vardhan Simhadri, Low Depth Cache-Oblivious Algorithms, 
SPAA '10 Proceedings of the 22nd ACM symposium on Parallelism in algorithms and architectures.
\bibitem{chan}T. Chan. Optimal output-sensitive convex hull algorithms in two and three dimensions. 
Discrete and Computational Geometry, 1996.
\bibitem{cole_vishkin} R. Cole and U. Vishkin. Approximate and Exact Parallel Scheduling with Application to List, Tree and 
Graph Problems. Proc. 27th IEEE Symposium on Foundation of Computer of Science, 1986. pp 478-491.
\bibitem{cole} R. Cole. Parallel Merge Sort. SIAM J. Compute. 17 (1988). pp 770-785.
\bibitem{resource}Richard Cole and Vijaya Ramachandran, Resource Oblivious Sorting on Multicores, 
ICALP'10 Proceedings of the 37th international colloquium conference on Automata, languages and programming, 2010.
\bibitem{refpoint} D. Dobkin and R.J. Lipton, Multidimensional searching problems, SIAM J. Comput., 5(1976), pp.181-186
\bibitem{reftranspose} Matteo Frigo, Charles E. Leiserson, Harald Prokop, Sridhar Ramachandran, Cache-Oblivious Algorithms, 
\bibitem{refreischuk} Rudiger Reischuk, A Fast Probabilistic Parallel Sorting Algorithm, Proc. 22nd IEEE Symposium on Foundation of Computer of Science, 
1981.
\bibitem{refsample}John H. Reif and Leslie G. Valiant, A Logarithmic Time Sort for Linear Size Networks, 
Journal of the Association for Computing Machinery, Vol. 34, No. I, January 1987, pp. 60-76.

Journal of the ACM, Vol. 49, No. 6, November 2002, pp. 828-858.
\bibitem{singlesort} Harald Prokop, Cache-Oblivious Algorithms, MS Thesis, MIT, June 1999.
IEEE 40th Annual Symposium on Foundations of Computer Science, 1999. 


\bibitem{3dhull}John H. Reif, Sandeep Sen, Optimal Parallel Randomized 
Algorithms for 3-D Convex Hulls and Related Problems (1992), SIAM 
Journal on Computing
\bibitem{RS:94} John H. Reif and Sandeep Sen,
  Randomized Algorithms for Binary Search and Load Balancing
               on Fixed Connection Networks with Geometric Applications,
  SIAM J. Comput., vol 23(3), 1994,pp.633-651.
\bibitem{ger}Random Sampling Techniques and Parallel Algorithms Design, Sandeep Sen, S. Rajasekaran, Synthesis of Parallel Algorithms,
J. Reif ed. , Morgan Kauffman, 1993.
\bibitem{RS:90} S. Rajasekaran and S. Sen, On parallel integer sorting, Acta 
Informatica, 29: 1--15, 1992.
\bibitem{chatterjee}Sandeep Sen, Sidhhartha Chatterjee and Neeraj Dumir, Towards a Theory of Cache-Efficient Algorithms, Journal of the ACM, 49(6), pp. 
828--858, 2002. 
\bibitem{sixteen} L. G. Valiant, A bridging model for multi-core computing, In Proc. of the 16th
Annual ESA, volume 5193, pages 13-28, 2008.


\end{thebibliography}
